\documentclass[12pt]{article}
\usepackage{amsmath}
\usepackage{graphicx}
\usepackage{enumerate}
\usepackage{natbib}
\newcommand{\blind}{1}

\usepackage{subcaption}
\usepackage{tikz}
\usetikzlibrary{positioning, arrows.meta, shapes.geometric}
\tikzset{%
  -Latex,semithick,
  >={Latex[width=1.5mm,length=2mm]},
  obs/.style n args = {2}{name = #1, circle, draw, inner sep = 7pt, label = center:$#2$},
  latent/.style n args = {2}{name = #1, circle, draw=gray, dashed, inner sep = 7pt, label = center:$#2$},
  greyDashed/.style={draw=gray, dashed}
}
\usepackage{hyperref}
\usepackage{amsfonts}
\usepackage{multirow}

\usepackage[ruled,vlined]{algorithm2e}
\usepackage{amsthm}
\newtheorem{theorem}{Theorem}
\newtheorem{prop}[theorem]{Proposition}
\newtheorem{definition}[theorem]{Definition}
\newtheorem{coro}[theorem]{Corollary}

\begin{document}

\if1\blind
{
  \title{\bf Feedback-augmented Non-homogeneous Hidden Markov Models for Longitudinal Causal Inference}
  \author{Jouni Helske\thanks{
   This work was supported by the INVEST Research Flagship Centre and the Research Council of Finland (decision numbers: 331817, 355153, 345546).}\hspace{.2cm}\\
    INVEST Research Flagship Centre, University of Turku}
  \maketitle
} \fi

\if0\blind
{
  \bigskip
  \bigskip
  \bigskip
  \begin{center}
    {\LARGE\bf Feedback-augmented Non-homogeneous Hidden Markov Models for Longitudinal Causal Inference}
\end{center}
  \medskip
} \fi

\bigskip
\begin{abstract}
Hidden Markov models are widely used for modeling sequential data but typically have limited applicability in observational causal inference due to their strong conditional independence assumptions. I introduce feedback-augmented non-homogeneous hidden Markov model (FAN-HMM), which incorporate time-varying covariates and feedback mechanisms from past observations to latent states and future responses. Integrating these models with the structural causal model framework allows flexible causal inference in longitudinal data with time-varying unobserved heterogeneity and multiple causal pathways. I show how, in a common case of categorical response variables, long-term causal effects can be estimated efficiently without the need for simulating counterfactual trajectories. Using simulation experiments, I study the performance of FAN-HMM under the common misspecification of the number of latent states, and finally apply the proposed approach to estimate the effect of the 2013 parental leave reform on fathers’ paternal leave uptake in Finnish workplaces.
\end{abstract}

\noindent%
\vfill

\newpage

\section{Introduction}

Discrete-state hidden Markov models (HMMs) and their variations are widely used to analyze temporal data in various settings, such as speech recognition \citep{Rabiner1989}, ecology \citep{Pohle2017, Glennie2023}, learning analytics \citep{Saqr2023, Helske2024_tutorial}, and life course research \citep{helske2018, Liao2022}. Standard homogeneous HMMs have been extended to cases where some or all model components (initial, transition, and emission distributions) depend on individual-specific, potentially time-varying covariates \citep{Hughes1994, MacDonald1997, Bartolucci2013}. \citet{Berchtold1999} introduced a double chain Markov model which extends the standard HMM by incorporating autoregressive dependencies at the observation level to model non-homogeneous time series. Similarly, \citet{Altman2007} proposed mixed-HMMs where both emission and transition distributions depend on individual-specific random effects and covariates, with observations belonging to the exponential family. In general, these and many other extensions can be classified as non-homogeneous HMM (NHMM), where some of model components vary over time and/or between individuals.

HMMs are commonly used in prediction tasks or as a method to extract and compress information of complex temporal data using relatively small number of interpretable latent states. Less work has been done on applying HMMs in the context of causal inference. \citet{Bartolucci2016} applied the potential outcomes (PO) framework and propensity scores \citep{Rosenbaum1983} to HMM setting where emission probabilities depend on latent states, which, in turn, depend on a discrete treatment variable. A similar approach, utilizing the concept of potential versions of discrete latent variables, was used by \citet{Bartolucci2023}, who considered the effect of a binary treatment variable on the transition probabilities of an HMM.

Recently, \citet{Helske2024_dmpm} discussed how standard HMMs are unsuitable for estimating certain type of long-term causal effects due to the conditional independence assumption of responses given the latent states and covariates. In this paper, I present the feedback-augmented NHMM (FAN-HMM), which extends standard homogeneous HMMs with time-varying covariates and allowing responses and hidden states to depend on previous responses. I show that estimating the causal effects of observables on responses using FAN-HMM can be done efficiently in case of categorical response variables without the need for simulating counterfactual trajectories as required in the Markovian models considered by \citet{Helske2024_dmpm}. Through simulation experiments, I examine how the common issue of misspecifying the number of hidden states affects both model estimation and the causal effect estimates. I then use FAN-HMM to estimate the effect of the 2013 parental leave reform in Finland on fathers' parental leave uptake within workplaces, where the hidden states capture the latent workplace culture influencing leave-taking probabilities.

The replication code for the paper is available on GitHub at \url{https://anonymous.4open.science/r/fanhmm/}.

\section{Models}

Let $y_{1:T} = y_1,\ldots,y_T$ be a sequence of observed symbols from a categorical distribution with $M$ different categories, and let $z_{1:T} = z_1,\ldots,z_T$ be a discrete-state first-order Markov chain with $S$ states. Instead of directly observing $z_{1:T}$, it is assumed that the conditional distribution of each observation $y_t$ depends on the latent state $z_t$, and given the latent states, the observations are conditionally independent, i.e., $P(y_t \mid y_{1:T}, z_{1:T}) = P(y_t \mid z_t)$. The time-homogeneous HMM is then defined by a set of parameters $\boldsymbol{\theta} = \{\boldsymbol{\pi}, \mathbf{A}, \mathbf{B}\}$, with initial state probabilities $\pi_s = P(z_1 = s)$, transition probabilities $a_{sr}= P(z_{t+1} = r \mid z_t = s)$, and emission probabilities $b_s(m) = P(y_t = m \mid z_t = s)$, $s,r=1,\ldots, S$, $m = 1,\ldots, M$. 

In general, the observations $y_t$ do not need to be categorical scalars; they can also be continuous or multivariate, potentially consisting of mixed distributions, with suitable modifications on the parameterization of the emission distributions. However, in this work, I focus on the important special case of categorical observations, which are commonly encountered for example in life course research where observations correspond variables such as employment or partnership status. 

When data contain multiple observation sequences $y_{t,i}$, $i=1,\ldots, N$, $t = 1, \ldots, T_i$, we can assume that each individual $i$ has their own latent state trajectory $z_{1:{T_i}}$, with both observations and latent states generated by the same $\boldsymbol{\theta}$. Therefore, this HMM is also homogeneous with respect to individuals. 

While these basic HMMs and their mixtures \citep{seqHMM} have been successfully applied to various settings, additional measurements of other, possibly time-varying, variables may influence the dynamics of observations and the underlying latent states. In such cases, we can generalize the HMM to a non-homogeneous HMM (NHMM), defined by ${\boldsymbol{\pi}_i, \mathbf{A}_{t,i}, \mathbf{B}_{t,i}}$. Here $\pi_i$ refers to initial hidden state probability vector for individual $i$, and $\mathbf{A}_{t,i}$ and $\mathbf{B}_{t,i}$ denote the transition and emission probability matrices at time $t$ for individual $i$, respectively (where $\mathbf{A}_{t,i}$ defines the transition from time $t-1$ to time $t$). These probabilities can depend on potentially different sets of covariate matrices $\mathbf{X}^\pi$, $\mathbf{X}^A$, and $\mathbf{X}^B$, which I assume contain no missing observations. For notational simplicity, I will occasionally omit the sequence index $i$. 

Let $\mathbf{x}^{\pi}_i$,$\mathbf{x}^A_{t,i}$, and $\mathbf{x}^B_{t,i}$ be vectors of lengths $K^\pi$, $K^A$, and $K^B$, corresponding to the covariates of sequence $i$ (at time $t$). Then, we can define $\boldsymbol{\pi}_i$ and the rows of $\mathbf{A}_{t,i}$ and $\mathbf{B}_{t,i}$ as
\begin{equation}\label{eq:softmax}
\boldsymbol{\pi}_i = \textrm{softmax}(\boldsymbol{\gamma}^\pi \mathbf{x}^{\pi}_i),\quad
\mathbf{A}_{t,i}(s,:) = \textrm{softmax}(\boldsymbol{\gamma}^A_s \mathbf{x}^A_{t,i}),\quad
\mathbf{B}_{t,i}(s,:) = \textrm{softmax}(\boldsymbol{\gamma}^B_s \mathbf{x}^B_{t,i}),
\end{equation}
for each sequence $i$, time point $t$, and current state $s = 1,\ldots, S$ at time $t$. Here $\boldsymbol{\gamma}^\pi$, $\boldsymbol{\gamma}^A_s$, and $\boldsymbol{\gamma}^B_s$ are $S \times K^{\pi}$, $S \times K^A$, and $M \times K^B$ matrices of regression coefficients. See Section A of the supplementary material for details on the parameterization of the coefficients ensuring the identifiability of the softmax function. 

For clarity of exposition, and without loss of generality, I now focus on a case where $x_t$ represents variable affecting both the observation $y_t$ and latent state $z_t$ (via emission and transition probabilities $\mathbf{B}_t$ and $\mathbf{A}_t$, respectively); variable $v_t$ affects only emission probabilities $\mathbf{B}_t$; and $w_t$ affect only transition probabilities $\mathbf{A}_t$. While initial state could depend on separate covariates, I assume for simplicity that $\boldsymbol{\pi}$ depends on $x_1$ and $w_1$. The directed acyclic graph (DAG) representing this model for three time points for a single sequence is shown in \autoref{fig:dagA}.

\begin{figure}[!ht]
\centering
\begin{subfigure}{0.45\textwidth}
\begin{tikzpicture}
    \node [obs = {w_1}{w_1}] at (0, 1.5) {$ \vphantom{0} $};
    \node [obs = {w_2}{w_2}] at (2, 1.5) {$ \vphantom{0} $};
    \node [obs = {w_3}{w_3}] at (4, 1.5) {$ \vphantom{0} $};
    \node [obs = {z_1}{z_1}] at (0, 3) {$ \vphantom{0} $};
    \node [obs = {z_2}{z_2}] at (2, 3) {$ \vphantom{0} $};
    \node [obs = {z_3}{z_3}] at (4, 3) {$ \vphantom{0} $};
    \node [obs = {y_1}{y_1}] at (0, 4.5) {$ \vphantom{0} $};
    \node [obs = {y_2}{y_2}] at (2, 4.5) {$ \vphantom{0} $};
    \node [obs = {y_3}{y_3}] at (4, 4.5) {$ \vphantom{0} $};
    \node [obs = {v_1}{v_1}] at (0, 6) {$ \vphantom{0} $};
    \node [obs = {v_2}{v_2}] at (2, 6) {$ \vphantom{0} $};
    \node [obs = {v_3}{v_3}] at (4, 6) {$ \vphantom{0} $};
    \node [obs = {x_1}{x_1}] at (0, 7.5) {$ \vphantom{0} $};
    \node [obs = {x_2}{x_2}] at (2, 7.5) {$ \vphantom{0} $};
    \node [obs = {x_3}{x_3}] at (4, 7.5) {$ \vphantom{0} $};
    \path [->] (x_3) edge [bend right = 45] (z_3);
  \path [->] (x_2) edge [bend right = 45] (z_2);
  \path [->] (w_3) edge (z_3);
  \path [->] (w_2) edge (z_2);
  \path [->] (w_1) edge (z_1);
  \path [->] (x_1) edge [bend right = 45] (z_1);
  \path [->] (v_3) edge (y_3);
  \path [->] (v_2) edge (y_2);
  \path [->] (v_1) edge (y_1);
  \path [->] (x_3) edge [bend right = 45] (y_3);
  \path [->] (x_2) edge [bend right = 45] (y_2);
  \path [->] (x_1) edge [bend right = 45] (y_1);
  \path [->] (z_1) edge (y_1);
  \path [->] (z_2) edge (y_2);
  \path [->] (z_3) edge (y_3);
  \path [->] (z_2) edge (z_3);
  \path [->] (z_1) edge (z_2);
\end{tikzpicture}
\caption{}
\label{fig:dagA}
\end{subfigure}
\hfill
\begin{subfigure}{0.45\textwidth}
\begin{tikzpicture}
    \node [obs = {w_1}{w_1}] at (0, 1.5) {$ \vphantom{0} $};
    \node [obs = {w_2}{w_2}] at (2, 1.5) {$ \vphantom{0} $};
    \node [obs = {w_3}{w_3}] at (4, 1.5) {$ \vphantom{0} $};
    \node [obs = {z_1}{z_1}] at (0, 3) {$ \vphantom{0} $};
    \node [obs = {z_2}{z_2}] at (2, 3) {$ \vphantom{0} $};
    \node [obs = {z_3}{z_3}] at (4, 3) {$ \vphantom{0} $};
    \node [obs = {y_1}{y_1}] at (0, 4.5) {$ \vphantom{0} $};
    \node [obs = {y_2}{y_2}] at (2, 4.5) {$ \vphantom{0} $};
    \node [obs = {y_3}{y_3}] at (4, 4.5) {$ \vphantom{0} $};
    \node [obs = {v_1}{v_1}] at (0, 6) {$ \vphantom{0} $};
    \node [obs = {v_2}{v_2}] at (2, 6) {$ \vphantom{0} $};
    \node [obs = {v_3}{v_3}] at (4, 6) {$ \vphantom{0} $};
    \node [obs = {x_1}{x_1}] at (0, 7.5) {$ \vphantom{0} $};
    \node [obs = {x_2}{x_2}] at (2, 7.5) {$ \vphantom{0} $};
    \node [obs = {x_3}{x_3}] at (4, 7.5) {$ \vphantom{0} $};
    \path [->] (x_3) edge [bend right = 45] (z_3);
  \path [->] (x_2) edge [bend right = 45] (z_2);
  \path [->] (w_3) edge (z_3);
  \path [->] (w_2) edge (z_2);
  \path [->] (w_1) edge (z_1);
  \path [->] (x_1) edge [bend right = 45] (z_1);
  \path [->] (y_2) edge (z_3);
  \path [->] (y_1) edge (z_2);
  \path [->] (y_2) edge (y_3);
  \path [->] (y_1) edge (y_2);
  \path [->] (v_3) edge (y_3);
  \path [->] (v_2) edge (y_2);
  \path [->] (v_1) edge (y_1);
  \path [->] (x_3) edge [bend right = 45] (y_3);
  \path [->] (x_2) edge [bend right = 45] (y_2);
  \path [->] (x_1) edge [bend right = 45] (y_1);
  \path [->] (z_1) edge (y_1);
  \path [->] (z_2) edge (y_2);
  \path [->] (z_3) edge (y_3);
  \path [->] (z_2) edge (z_3);
  \path [->] (z_1) edge (z_2);
\end{tikzpicture}
\caption{}
\label{fig:dagB}
\end{subfigure}
\caption{Directed acyclic graphs of NHMM (a) and FAN-HMM (b) for time points $t=1,2,3$. In (a), state $z$ depends on past state and covariates $x$ and $w$, whereas observation $y$ depends on current state and covariates $x$ and $v$. In (b), the transition and emission probabilities depend also on past observations.}
    \label{fig:nhmms}
\end{figure}
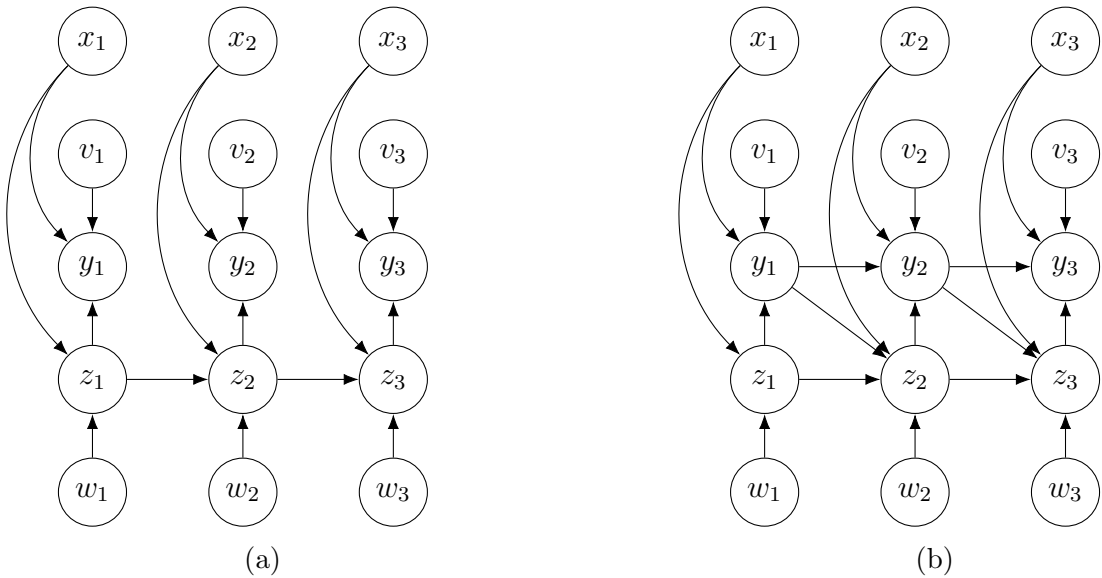

On \autoref{fig:dagB}, NHMM is extended to FAN-NHMM, where an additional edges $y_{t-1}\to y_t$ induce direct observational-level autocorrelation, and edges $y_{t-1} \to z_t$ enablefeedback from the observational level back to the  latent level.

As an example of the direct effect of $y_{t-1} \to y_t$ and feedback $y_{t-1} \to z_t$, consider a study on siblings' enrollment to higher education in a developing country with large family sizes. Here, the time variable represents the child's rank in birth order. A child's educational choices can have a direct effect on their younger sibling ($y_t \to y_{t+1})$, but these choices can also influence the latent state that represents the general expectations and capabilities of the family regarding education ($y_t \to z_{t+1}$). This, in turn, affects the choices of the younger siblings ($z_{t+1} \to y_{t+1}$).
 
While missing response variables are straightforward to handle in HMM and NHMM, they need to be explicitly marginalized out in the forward algorithm \citep{Rabiner1989} of FAN-HMM due to the dependency of $\mathbf{A}$ and $\mathbf{B}$ on responses. In complete data case, FAN-HMM can be formulated as an NHMM using Equations \ref{eq:softmax} by including relevant (lagged) response variables to the covariate matrices $\mathbf{X}^A$ and $\mathbf{X}^B$, and treating the first observation $y_1$ as fixed (or by defining $\mathbf{B}_1$ separately without dependency on $y_0$), so parameter estimation of FAN-HMMs pose no additional challenges compared to NHMMs. However, the dependency of states and observations on past observations must still be accounted for in forward prediction of counterfactual trajectories for causal inference.

For parameter estimation, I follow the maximum likelihood approach. As with other mixture models, the likelihood of HMMs is known to be highly multimodal \citep{Rabiner1989, Cappe2005, Zucchini2009, Visser2022}. In maximum likelihood estimation, local optima and flat regions of the likelihood surface can lead numerical optimization algorithms to converge suboptimal solutions. Therefore, the model is typically estimated multiple times using different initial values,  so we can then claim with some certainty that the best solution corresponds to the global optimum. Section B of the supplementary material contains further discussion on the parameter estimation of HMMs, also from a Bayesian perspective.

\section{Causal inference using with FAN-HMMs}\label{sec:causality}

Instead of PO framework used, e.g., by \citet{Bartolucci2016, Bartolucci2023}, I consider the causal inference within the framework of structural causal models (SCMs) \citep{Pearl2009}, which focus on the effects of hypothetical interventions that alter the functional structure of the system under study. Denoting the intervention/treatment variable as $X$ and the outcome variable as $Y$, an intervention on $X$, denoted by $\operatorname{do}(X = x)$, leads to an interventional distribution $p(Y = y| \operatorname{do}(X = x))$, abbreviated as $p(y \mid \operatorname{do}(x))$. If $p(y \mid \operatorname{do}(x))$ can be derived from the observed probabilities of the SCM, the causal effect of $X$ on $Y$ is said to be \emph{identifiable}. 

Under the SCM approach, the causal assumptions of the system can be encoded as a DAG, which aligns well with the HMMs which are often represented using DAGs even in non-causal settings. Because the sufficient and necessary conditions for identifiability are encoded in this DAG, we can use graph-based methods such as backdoor criterion and general do-calculus \citep{Pearl2009} to assess the identifiability of the causal effects.

\autoref{fig:causaldag} presents two causal graphs that generalize the graph in \autoref{fig:dagB}. In both graphs, latent variables $u$, induce dependency between all three covariates over time, while in in \autoref{fig:fanhmmB} there is an additional direct effect of $x_t \to x_{t+1}$ and similarly for $v$ and $w$ variables. To simplify notation, without loss of generality, I continue to treat the covariate nodes $x_t$, $w_t$, and $v_t$ in \autoref{fig:causaldag} as univariate, although they can generally consist of multiple variables \citep{Tikka2023}. The dashed edges and nodes reflect variables and dependencies which are not explicitly modeled by FAN-HMM, yet they must be accounted for when deriving the identifying functionals of causal effects.

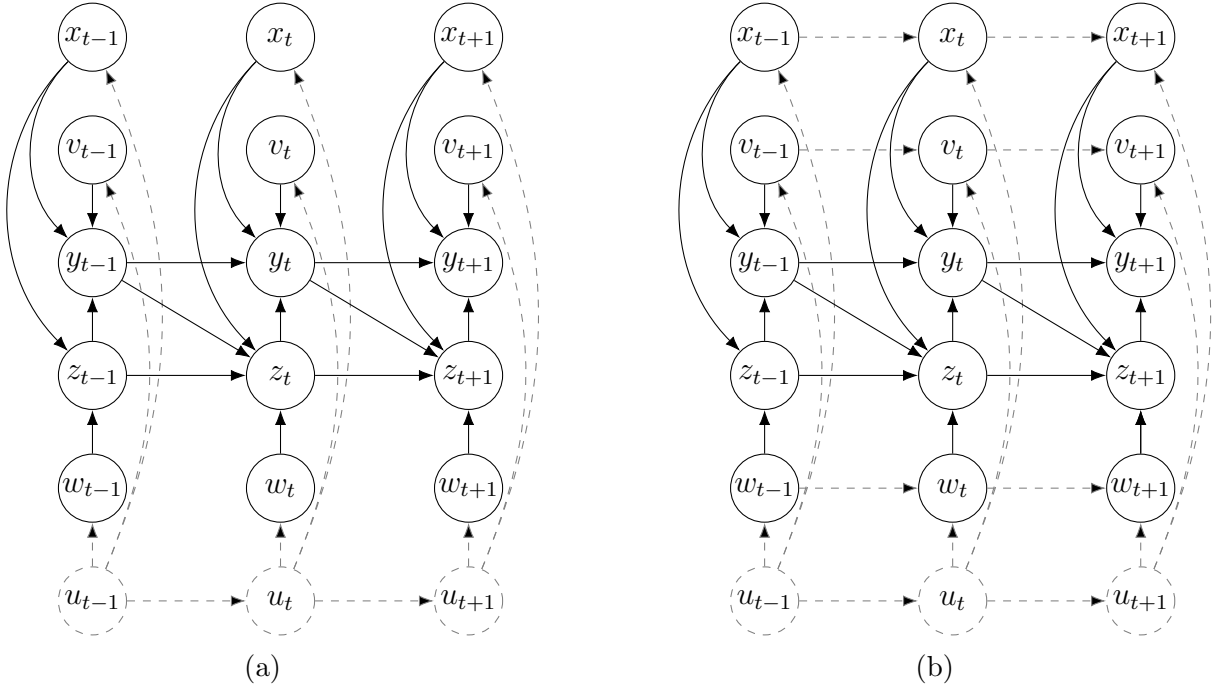
\begin{figure}[!ht]
\centering
\begin{subfigure}{0.45\textwidth}
\begin{tikzpicture}
    \node [latent = {u_{t-1}}{u_{t-1}}] at (0, 0) {$ \vphantom{0} $};
    \node [latent = {u_t}{u_t}] at (2.5, 0) {$ \vphantom{0} $};
    \node [latent = {u_{t+1}}{u_{t+1}}] at (5, 0) {$ \vphantom{0} $};
    \node [obs = {w_{t-1}}{w_{t-1}}] at (0, 1.5) {$ \vphantom{0} $};
    \node [obs = {w_t}{w_t}] at (2.5, 1.5) {$ \vphantom{0} $};
    \node [obs = {w_{t+1}}{w_{t+1}}] at (5, 1.5) {$ \vphantom{0} $};
    \node [obs = {z_{t-1}}{z_{t-1}}] at (0, 3) {$ \vphantom{0} $};
    \node [obs = {z_t}{z_t}] at (2.5, 3) {$ \vphantom{0} $};
    \node [obs = {z_{t+1}}{z_{t+1}}] at (5, 3) {$ \vphantom{0} $};
    \node [obs = {y_{t-1}}{y_{t-1}}] at (0, 4.5) {$ \vphantom{0} $};
    \node [obs = {y_t}{y_t}] at (2.5, 4.5) {$ \vphantom{0} $};
    \node [obs = {y_{t+1}}{y_{t+1}}] at (5, 4.5) {$ \vphantom{0} $};
    \node [obs = {v_{t-1}}{v_{t-1}}] at (0, 6) {$ \vphantom{0} $};
    \node [obs = {v_t}{v_t}] at (2.5, 6) {$ \vphantom{0} $};
    \node [obs = {v_{t+1}}{v_{t+1}}] at (5, 6) {$ \vphantom{0} $};
    \node [obs = {x_{t-1}}{x_{t-1}}] at (0, 7.5) {$ \vphantom{0} $};
    \node [obs = {x_t}{x_t}] at (2.5, 7.5) {$ \vphantom{0} $};
    \node [obs = {x_{t+1}}{x_{t+1}}] at (5, 7.5) {$ \vphantom{0} $};
    \path [->, greyDashed] (u_t) edge (u_{t+1});
    \path [->, greyDashed] (u_{t-1}) edge (u_t);
    \path [->, greyDashed] (u_{t+1}) edge [bend right = 22.5] (x_{t+1});
    \path [->, greyDashed] (u_{t+1}) edge [bend right = 22.5] (v_{t+1});
    \path [->, greyDashed] (u_t) edge [bend right = 22.5] (x_t);
    \path [->, greyDashed] (u_t) edge [bend right = 22.5] (v_t);
    \path [->, greyDashed] (u_{t+1}) edge (w_{t+1});
    \path [->, greyDashed] (u_t) edge (w_t);
    \path [->, greyDashed] (u_{t-1}) edge [bend right = 22.5] (x_{t-1});
    \path [->, greyDashed] (u_{t-1}) edge [bend right = 22.5] (v_{t-1});
    \path [->, greyDashed] (u_{t-1}) edge (w_{t-1});
    \path [->] (v_{t+1}) edge (y_{t+1});
    \path [->] (v_t) edge (y_t);
    \path [->] (v_{t-1}) edge (y_{t-1});
    \path [->] (w_{t-1}) edge (z_{t-1});
    \path [->] (w_t) edge (z_t);
    \path [->] (w_{t+1}) edge (z_{t+1});
     \path [->] (x_{t-1}) edge [bend right = 45] (z_{t-1});
     \path [->] (x_t) edge [bend right = 45] (z_t);
     \path [->] (x_{t+1}) edge [bend right = 45] (z_{t+1});
    \path [->] (x_{t+1}) edge [bend right = 45] (y_{t+1});
    \path [->] (x_t) edge [bend right = 45] (y_t);
    \path [->] (x_{t-1}) edge [bend right = 45] (y_{t-1});
    \path [->] (z_{t-1}) edge (y_{t-1});
    \path [->] (z_t) edge (y_t);
    \path [->] (z_{t+1}) edge (y_{t+1});
    \path [->] (z_t) edge (z_{t+1});
    \path [->] (z_{t-1}) edge (z_t);
    \path [->] (y_{t-1}) edge (y_t);
    \path [->] (y_t) edge (y_{t+1});
    \path [->] (y_{t-1}) edge (z_t);
    \path [->] (y_t) edge (z_{t+1});
\end{tikzpicture}
\caption{}
\label{fig:fanhmmA}
\end{subfigure}
\hfill
\begin{subfigure}{0.45\textwidth}
\begin{tikzpicture}
       \node [latent = {u_{t-1}}{u_{t-1}}] at (0, 0) {$ \vphantom{0} $};
    \node [latent = {u_t}{u_t}] at (2.5, 0) {$ \vphantom{0} $};
    \node [latent = {u_{t+1}}{u_{t+1}}] at (5, 0) {$ \vphantom{0} $};
    \node [obs = {w_{t-1}}{w_{t-1}}] at (0, 1.5) {$ \vphantom{0} $};
    \node [obs = {w_t}{w_t}] at (2.5, 1.5) {$ \vphantom{0} $};
    \node [obs = {w_{t+1}}{w_{t+1}}] at (5, 1.5) {$ \vphantom{0} $};
    \node [obs = {z_{t-1}}{z_{t-1}}] at (0, 3) {$ \vphantom{0} $};
    \node [obs = {z_t}{z_t}] at (2.5, 3) {$ \vphantom{0} $};
    \node [obs = {z_{t+1}}{z_{t+1}}] at (5, 3) {$ \vphantom{0} $};
    \node [obs = {y_{t-1}}{y_{t-1}}] at (0, 4.5) {$ \vphantom{0} $};
    \node [obs = {y_t}{y_t}] at (2.5, 4.5) {$ \vphantom{0} $};
    \node [obs = {y_{t+1}}{y_{t+1}}] at (5, 4.5) {$ \vphantom{0} $};
    \node [obs = {v_{t-1}}{v_{t-1}}] at (0, 6) {$ \vphantom{0} $};
    \node [obs = {v_t}{v_t}] at (2.5, 6) {$ \vphantom{0} $};
    \node [obs = {v_{t+1}}{v_{t+1}}] at (5, 6) {$ \vphantom{0} $};
    \node [obs = {x_{t-1}}{x_{t-1}}] at (0, 7.5) {$ \vphantom{0} $};
    \node [obs = {x_t}{x_t}] at (2.5, 7.5) {$ \vphantom{0} $};
    \node [obs = {x_{t+1}}{x_{t+1}}] at (5, 7.5) {$ \vphantom{0} $};
    \path [->, greyDashed] (u_t) edge (u_{t+1});
    \path [->, greyDashed] (u_{t-1}) edge (u_t);
    \path [->, greyDashed] (u_{t+1}) edge [bend right = 22.5] (x_{t+1});
    \path [->, greyDashed] (u_{t+1}) edge [bend right = 22.5] (v_{t+1});
    \path [->, greyDashed] (u_t) edge [bend right = 22.5] (x_t);
    \path [->, greyDashed] (u_t) edge [bend right = 22.5] (v_t);
    \path [->, greyDashed] (u_{t+1}) edge (w_{t+1});
    \path [->, greyDashed] (u_t) edge (w_t);
    \path [->, greyDashed] (u_{t-1}) edge [bend right = 22.5] (x_{t-1});
    \path [->, greyDashed] (u_{t-1}) edge [bend right = 22.5] (v_{t-1});
    \path [->, greyDashed] (u_{t-1}) edge (w_{t-1});
    \path [->] (v_{t+1}) edge (y_{t+1});
    \path [->] (v_t) edge (y_t);
    \path [->] (v_{t-1}) edge (y_{t-1});
    \path [->] (w_{t-1}) edge (z_{t-1});
    \path [->] (w_t) edge (z_t);
    \path [->] (w_{t+1}) edge (z_{t+1});
    \path [->] (w_{t+1}) edge (z_{t+1});
     \path [->] (x_{t-1}) edge [bend right = 45] (z_{t-1});
     \path [->] (x_t) edge [bend right = 45] (z_t);
     \path [->] (x_{t+1}) edge [bend right = 45] (z_{t+1});
    \path [->] (x_{t+1}) edge [bend right = 45] (y_{t+1});
    \path [->] (x_t) edge [bend right = 45] (y_t);
    \path [->] (x_{t-1}) edge [bend right = 45] (y_{t-1});
    \path [->] (z_{t-1}) edge (y_{t-1});
    \path [->] (z_t) edge (y_t);
    \path [->] (z_{t+1}) edge (y_{t+1});
    \path [->] (z_t) edge (z_{t+1});
    \path [->] (z_{t-1}) edge (z_t);
    \path [->] (y_{t-1}) edge (y_t);
    \path [->] (y_t) edge (y_{t+1});
    \path [->] (y_{t-1}) edge (z_t);
    \path [->] (y_t) edge (z_{t+1});
    \path [->, greyDashed] (x_{t-1}) edge (x_t);
    \path [->, greyDashed] (x_t) edge (x_{t+1});    
    \path [->, greyDashed] (v_{t-1}) edge (v_t);
    \path [->, greyDashed] (v_t) edge (v_{t+1});
    \path [->, greyDashed] (w_{t-1}) edge (w_t);
    \path [->, greyDashed] (w_t) edge (w_{t+1});
\end{tikzpicture}
\caption{}
\label{fig:fanhmmB}
\end{subfigure}
    \caption{Two causal graphs, with different assumptions about the dependency structure between the covariates (represented by gray dashed arrows and nodes). On the left, the unobserved variable $u_t$ has a direct causal effect on covariates $x_t$, $v_t$, and $w_t$. On the right, each covariate also directly depends on its own past values.}
    \label{fig:causaldag}
\end{figure}

The role of the latent variables $z$ can vary depending on the objectives of the causal inference task. First, $z$ can be considered as an auxiliary variable that captures time-varying individual-level heterogeneity without being of direct interest. In this case, our causal queries are of the form $p(y \mid \operatorname{do}(x))$. Alternatively, we may be interested in conditional causal effects of the form $p(y \mid \operatorname{do}(x), z)$, where the focus is in how the effects of $x$ vary with $z$. In the first case, since the latent states are marginalized out, selecting the correct number of states and interpreting them may be less critical than in the second case, which allows us to examine how potential time-varying effects of interventions depend on the changes in population composition with respect to latent (time-varying) groups characterized by $z$.

To establish the necessary background, Definition \ref{def:basics} lists key concepts related to causal graphs, which will be used in the following results.

\begin{definition}\label{def:basics}
Let $\mathcal{G}$ be a directed acyclic graph. Then,
\begin{enumerate} 
\item A \emph{back-door path} from $x$ to $y$ is any path in $\mathcal{G}$ from $x$ to $y$ that contains an incoming edge to $x$. 
\item A node $z$ is a \emph{collider} in path $P$ if $P$ contains a pattern $a \to z \leftarrow b$.
\item A set of nodes $C$ \emph{blocks} a path $P$ if either $P$ contains a non-collider that is in $C$, or $P$ contains a collider that is not in $C$ and has no descendants in $C$. If $P$ is not blocked, it is said to be \emph{open}.
\item A set $C$ is a \emph{valid adjustment set} for the causal effect $p(y \mid \operatorname{do}(x))$ if (i) no nodes in $C$ are descendants of $x$, and (ii) $C$ blocks every back-door path from $x$ to $y$.  
\end{enumerate}  
\end{definition}  

\begin{theorem}[Back-door criterion \citet{Pearl1995}]\label{th:backdoor}
Given a DAG $\mathcal{G}$, if $C$ is a valid adjustment set relative to intervention $x$ and outcome $y$, then the causal effect of $x$ on $y$ is identifiable and is given by the back-door adjustment formula
\begin{equation}\label{eq:backdoor}
p(y \mid \operatorname{do}(x)) = \sum_{C} p(y \mid x, C)p(C).  
\end{equation}  
\end{theorem}  

In the following, let $S_t = \{y_{1:t}, x_{1:t}, v_{1:t}, w_{1:t}\}$, and denote $p(z_{t} \mid x_{t}, w_{t}, z_{t-1}, y_{t-1})$ as $a(z_{t})$ and $p(y_{t} \mid z_{t}, x_{t}, v_{t}, y_{t-1})$ as $b(y_{t})$ for short. Proposition \ref{prop:immediate} considers a case where the interest is in the immediate effect of intervention at time $t$ on outcome $y_t$.

\begin{prop}[Immediate causal effects]\label{prop:immediate}
Let $\mathcal{G}$ be the causal graph defined in either \autoref{fig:fanhmmA} or \autoref{fig:fanhmmB}, and define $C=\{v_t, w_t, z_{t-1}, S_{t-1}\}$. Then, the causal effect (i.e., interventional distribution) $p(y_t \mid \operatorname{do}(x_t))$ is identifiable from $\mathcal{G}$ as
\begin{equation}\label{eq:dox}
\begin{aligned}
p(y_t\mid \operatorname{do}(x_t)) =& \sum_{C} p(y_t \mid x_t, C)p(C)\\
=& \sum_{C, z_t} p(y_t \mid z_t, x_t, v_t, y_{t-1})p(z_t \mid x_t, w_t, z_{t-1}, y_{t-1})p( v_t, w_t, z_{t-1},S_{t-1}),
\end{aligned}
\end{equation}
while $p(y_t \mid \operatorname{do}(v_t))$ and $p(y_t \mid \operatorname{do}(w_t))$ are obtained by replacing $x_t$ with $v_t$ and $w_t$, respectively, in $C$ and \autoref{eq:dox}.
\end{prop}
\begin{proof}
By Definition \ref{def:basics}, set $\{v_t, w_t, z_{t-1}, y_{t-1}\} \in C$ is a valid adjustment set in $\mathcal{G}$, and adding additional ancestors of these nodes to $C$ does not violate the conditions (i) and (ii). Therefore $C$ is valid adjustment set and back-door adjustment of Theorem \ref{th:backdoor} is applicable. 
The second equality in \autoref{eq:dox} follows from standard conditional independence relations in $\mathcal{G}$ and marginalization of $z_t$. The same reasoning applies analogously to $p(y_t \mid \operatorname{do}(v_t))$ and $p(y_t \mid \operatorname{do}(w_t))$.
\end{proof}

The reason for adjusting on the whole past $S_{t-1}$ will become evident in \autoref{sec:causalestimation}. The following proposition generalizes Proposition \ref{prop:immediate} to a case of where the interest is in the joint distribution of $y_{t+k}$ and $z_{t+k}$ after recurring interventions on $x_{t:t+k}$. The results for $\operatorname{do}(v)$ and $\operatorname{do}(w)$ are analogous. 

\begin{prop}[Long term causal effects of recurring interventions]\label{prop:recurring}
Let $\mathcal{G}$ be the causal graph defined in either \autoref{fig:fanhmmA} or \autoref{fig:fanhmmB}, $k \geq 0$, and $C = \{v_{t:t+k}, w_{t:t+k}, z_{t-1}, S_{t-1}\}$. Then the causal effect $p(y_{t+k}, z_{t+k} \mid \operatorname{do}(x_{t:t+k}))$ is identifiable from $\mathcal{G}$ as
\begin{equation}\label{eq:doxt}
\begin{aligned}
p(y_{t+k}, z_{t+k} \mid \operatorname{do}(x_{t:t+k})) =& \sum_{C}  p(y_{t+k}, z_{t+k} \mid x_{t:t+k}, C) p(C)\\
=& \sum_{C} \left[\sum_{\substack{y_{t:t+k-1},\\z_{t:t+k-1}}} \prod_{j=t}^{t+k} b(y_j) a(z_j)\right] p(C).
\end{aligned}
\end{equation}
Furthermore,
\begin{equation}
p(y_{t+k} \mid \operatorname{do}(x_{t:t+k})) = \sum_{z_{t+k}} p(y_{t+k}, z_{t+k} \mid \operatorname{do}(x_{t:t+k})),
\end{equation}
and
\begin{equation}
p(y_{t+k} \mid \operatorname{do}(x_{t:t+k}), z_{t+k}) = \frac{p(y_{t+k}, z_{t+k} \mid \operatorname{do}(x_{t:t+k}))}{\sum_{y_{t+k}} p(y_{t+k}, z_{t+k} \mid \operatorname{do}(x_{t:t+k}))}.
\end{equation}
\end{prop}
\begin{proof}
As in the proof of \ref{prop:immediate}, we verify that the set $C$ satisfies the back-door criterion: (i) None of the nodes $x_{t:t+k}$ are ancestors of any nodes of $C$. (ii) Nodes $w_{t:t+k}$ and $v_{t:t+k}$ block all back-door paths from $x_{t:t+k}$ to $y_{t+k}$ and $z_{t+k}$ that pass through $u_{t:t+k}$, while $y_{t-1}$ and $z_{t-1}$ block the remaining back-doors paths that go through $u_{t-1}$. Additional variables in $S_{t-1}$ are not descendants of the intervention variables and do not open any back-doors paths. The second equality of \autoref{eq:doxt} follows from factorization of $p(y_{t+k} \mid x_{t:t+k}, C)$ and marginalization of intermediate responses $y_{t:t+k-1}$ and hidden states $z_{t:t+k-1}$. 
\end{proof}

Note that while the values of $x_{t:t+k}$ in $p(y_{t+k}, z_{t+k} \mid \operatorname{do}(x_{t:t+k}))$ can vary (e.g., $\operatorname{do}(x_t = 0, x_{t+1}=1)$), this sequence of interventions is fixed at time $t$. That is, the intermediate response values $y_t$ cannot affect subsequent interventions $\operatorname{do}(x_{t+1})$. This differs from so-called dynamic treatment regimes \citep{Murphy2003} where at each time point the intervention is adjusted based on the current information. 

While the results so far hold for both DAGs of \autoref{fig:causaldag}, the following corollary demonstrates how direct effect of $x_{t-1}$ on $x_t$ affects the resulting adjustment formula.
 
\begin{coro}[Long-term causal effects of atomic interventions]\label{prop:atomic}
Let $\mathcal{G}_1$ be the causal graph defined in \autoref{fig:fanhmmA} and $\mathcal{G}_2$ be \autoref{fig:fanhmmB}. Then $p(y_{t+k} \mid \operatorname{do}(x_t))$ is identifiable from $\mathcal{G}_1$ as
\begin{equation}\label{eq:doxt1}
\begin{aligned}
p(y_{t+k} \mid \operatorname{do}(x_{t})) =& \sum_{C_1}  p(y_{t+k} \mid x_{t}, C_1) p(C_1)\\
=& \sum_{C_1} \left[\sum_{\substack{y_{t:t+k-1},\\z_{t:t+k}}} \prod_{j=t}^{t+k} b(y_j) a(z_j)\right] p(C_1),
\end{aligned}
\end{equation}
where $C_1 = \{x_{t+1:t+k}, w_{t:t+k}, v_{t:t+k}, z_{t-1}, S_{t-1}\}$. 
For $\mathcal{G}_2$,
\begin{equation}\label{eq:doxt2}
\begin{aligned}
p(y_{t+k} \mid \operatorname{do}(x_{t})) =& \sum_{C_2}  p(y_{t+k} \mid x_{t}, C_2) p(C_2)\\
=& \sum_{C_2} \left[\sum_{\substack{y_{t:t+k-1},\\z_{t:t+k},\\x_{t+1:t+k}}} \prod_{j=t}^{t+k} b(y_j) a(z_j)\prod_{j=t+1}^{t+k}p(x_{j} \mid x_{j-1})\right] p(C_2),
\end{aligned}
\end{equation}
where $C_2 = \{w_{t:t+k}, v_{t:t+k}, z_{t-1}, S_{t-1}\}$. 
\end{coro}
\begin{proof}
The result follows directly from Proposition \ref{prop:recurring}.
\end{proof}

Corollary \ref{prop:atomic} shows that when there are no direct dependencies between covariates over time, the effects of atomic interventions can be identified as in the same way as recurring interventions, simply by including the intermediate values $x_{t+1:t+k}$ in the adjustment set. However, when such dependencies are present (as in \autoref{fig:fanhmmB}), intervening on $x_t$ affects also $x_{t+1}$, which necessitates modeling also the condtional distribution $p(x_{t+1} | x_t)$. Thus, in this case we can no longer operate solely under the FAN-HMMs considered here.

Naturally the graphs of \autoref{fig:causaldag} are just two possible DAGs depicting the true data-generating process (DGP). However in practice, they are likely more realistic than assuming that $x_t$, $v_t$ and $w_t$ are completely independent, both mutually and temporally. Importantly, assuming either of the DAGs in \autoref{fig:causaldag} does not bias the causal effect estimates if the true DGP corresponds to the graph without the dashed components. This also holds if the DGP follows an NHMM without outgoing arrows from $y_t$.  In summary, while it is always possible to perform predictions of $y_t$ given the current or past values of the model covariates, domain knowledge about the underlying causal graph is essential when assessing whether such predictions can be interpreted causally. 

Here the focus was on causal effects of observables, i.e., effects of covariates (including past responses) on the responses, but technically we could also be consider cases where the variable targeted by the intervention is a latent state. However, because the interpretation of the hidden states depends on the (time-varying) transition and emission probabilities typically estimated from the data, it is in general challenging conceptualize the meaning of such interventions, unless one has very strong prior knowledge on the number of hidden states and their structure. \citet{Zucchini2009} warn against over-interpreting hidden Markov models and latent variable models in general, an issue which should be kept in mind if first inferring the hidden states from the data, and subsequently conducting causal inference on those states.

\subsection{Estimating the causal effects under FAN-HMM}\label{sec:causalestimation}

I now consider the practical estimation of the causal effects discussed in \autoref{sec:causality}. Throughout this section, I assume that we have data on $y_t, x_t, v_t,w_t$, for $t=1,\ldots,T$ and $N$ individuals, generated by a FAN-HMM with a known number of hidden states and model parameters $\boldsymbol{\gamma}$ as given in \autoref{eq:softmax}. 

Instead of explicitly modelling the distribution of the adjustment set $C$, it is typically approximated by the empirical distribution of the data \citep{Hernan2020}. However, in the context of NHMMs, $C$ typically includes latent variable $z$. Therefore, to use the results from \autoref{sec:causality}, we need to compute the conditional distributions of the latent states while, in case of long-term effects, marginalizing over the intermediate states and responses. 

Denote transition and emission matrices at time $t$ given observation $y_{t-1}$ as $\mathbf{A}_t(y_{t-1})$ and $\mathbf{B}_t(y_{t-1})$. If the model does not contain edge $y_{t-1} \to z_t$, then $\mathbf{A}_t(y_{t-1}) = \mathbf{A}_t$, and similarly for emissions if $y_{t-1} \not\to y_t$. Furthermore, let $\alpha_{t}=p(z_t, y_{1:t})$ and $\bar \alpha_{t}=p(z_t \mid y_{1:t})$ be the unnormalized and normalized forward variables, where the dependency on covariates $x, v, w$ is implicit. Additionally, let $\mathbf{e}_{i}$ be the basis vector of length $M$ where $e(j) = 1$ if $j=i$ and $0$ otherwise, and $\mathbf{1}_n = (1,\ldots,1)^\top \in \mathbb{R}^n$. Based on this notation, Algorithm \ref{alg:forward} defines the forward algorithm for FAN-HMM, which outputs $\mathbf{D}_t(m, s) = p(y_t = m, z_t = s \mid S_t\setminus y_t)$. By modifying the relevant transition and emission matrices and computing $\mathbf{D}_t$ for every sequence, we get
\begin{equation}\label{eq:est_doxt}
\begin{aligned}
p(y_{t}, z_{t} \mid \operatorname{do}(x_t)) &= 
 \sum_{C}  p(y_t, z_t \mid x_t, C) p(C)\\
&\approx \frac{1}{N}\sum_{i=1}^N \mathbf{D}_{t, i},
\end{aligned}
\end{equation}
where $C = \{v_{t}, w_{t}, z_{t-1}, S_{t-1}\}$. The long-term effects are obtained similarly by setting $y_{t:t+k-1}$ missing in Algorithm \ref{alg:forward}, giving us $\mathbf{D}_{t+k}(m, s) = p(y_{t+k} = m, z_{t+k} = s \mid  S_{t+k} \setminus \{y_{t:t+k-1}\})$. The marginals $p(y_{t+k} \mid \operatorname{do}(x_{t:t+k}))$ and conditionals $p(y_{t+k} \mid \operatorname{do}(x_{t:t+k}), z_{t+k} = s)$ are then obtained as (scaled) row sums of the resulting matrix $\mathbf{D}_{t+k}$. Therefore, in the case of categorical response variables, we can estimate various short- and long-term causal effects accurately and computationally efficient way without needing to rely on for example Monte Carlo simulation employed in dynamic multivariate panel models of \citet{Helske2024_dmpm}. For uncertainty intervals, the nonparametric bootstrap of individuals yields samples of the model coefficients, which can be used to compute \autoref{eq:est_doxt} repeatedly for bootstrap samples of $p(y_{t}, z_{t} \mid \operatorname{do}(x_t))$. The uncertainty of $P(C)$ can also be accounted for by using the bootstrapped data in marginalization.

\begin{algorithm}[!ht]
\caption{Forward algorithm for FAN-HMM}\label{alg:forward}
\renewcommand{\gets}{=}
\DontPrintSemicolon
\eIf{FAN-HMM has an edge $y_t \to y_{t+1}$}{
    $\mathbf{D}_1 \gets \boldsymbol{\pi} \boldsymbol{e}_{y_1}^\top$\;
}{
    $\boldsymbol{D}_1 \gets \operatorname{diag}(\boldsymbol{\pi})\mathbf{B}_1$\;
}
\eIf{$y_1$ is observed}{
    $\boldsymbol{\alpha}_1 \gets \boldsymbol{D}_1 \boldsymbol{e}_{y_1}$\; 
}{
    $\boldsymbol{\alpha}_1 \gets \boldsymbol{D}_1 \boldsymbol{1_M}$\;
}
$\boldsymbol{\bar\alpha}_1 \gets \boldsymbol{\alpha}_1 / \sum \boldsymbol{\alpha}_1$\;
\For{$t = 2$ \KwTo $T$}{
    \eIf{$y_{t-1}$ is observed}{
        $\boldsymbol{D}_t \gets \operatorname{diag}(\mathbf{A}_t^\top(y_{t-1}) \boldsymbol{\bar\alpha}_{t-1})\mathbf{B}_t(y_{t-1})$\;
    }{
        $\boldsymbol{D}_t \gets \sum_{m=0}^M \operatorname{diag}(\mathbf{A}_t^\top(m) \boldsymbol{\bar\alpha}_{t-1}) \mathbf{B}_t(m) \boldsymbol{1^\top_S}\boldsymbol{D}_{t-1} \boldsymbol{e}_m$\;
    }
    \eIf{$y_t$ is observed}{
        $\boldsymbol{\alpha}_t \gets \boldsymbol{D}_t \boldsymbol{e}_{y_t}$\;
    }{
        $\boldsymbol{\alpha}_t \gets \boldsymbol{D}_t \boldsymbol{1_M}$\;
    }
    $\boldsymbol{\bar\alpha}_t \gets \boldsymbol{\alpha}_t / \sum \boldsymbol{\alpha}_t$\;
}
\end{algorithm}

\section{Simulation experiments}\label{sec:experiments}

I first study the parameter estimation of FAN-HMM, and how multimodality of the likelihood depends on the potential model misspecification, in particular the incorrect specification of the number of hidden states. Section D of the supplementary materials contains further simulation experiments. 

All computations were done in \texttt{R} \citep{rcore}, using the \texttt{seqHMM} \citep{seqHMM}, \texttt{ggplot2} \citep{ggplot2} and \texttt{dplyr} \citep{dplyr} packages.

I simulated synthetic data from FAN-HMM with the true number of hidden states $S=3$, number of observed categories $M = 4$, number of time points $T = 50$, and the number of sequences $N = 1000$. In addition to autoregressive effect of previous observation $y_{t-1}$ on emission probabilities $\mathbf{B}_t$, and feedback effect of current observation $y_t$ to the transition probability matrix $\mathbf{A}_t$, both $\mathbf{A}_t$ and $\mathbf{B}_t$ depend on covariate $x$ simulated as 
\begin{equation}\label{eq:covariates}
\begin{aligned}
x_{i, t} \sim \textrm{N}(u_i + v_i t, 0.1^2),\quad u_{i} \sim \textrm{N}(0, 0.5^2), \quad v_{i} \sim \textrm{Uniform}(-0.05, 0.05).
\end{aligned}
\end{equation}

I fixed the $\eta$ coefficients (see Section A of the supplementary material) corresponding to the intercepts of $\boldsymbol{\pi}$, $\mathbf{A}$, and $\mathbf{B}$ to match the average initial, transition and emission probabilities defined as
\begin{equation}\label{eq:fanhmm_avg}
\pi = \begin{pmatrix}
     0.8  \\
     0.1 \\
     0.1
    \end{pmatrix},
A = \begin{pmatrix}
     & 0.85  & 0.1 & 0.05 \\
     & 0.1 & 0.8 & 0.1\\
     & 0.05 & 0.05 & 0.9
    \end{pmatrix},
B = \begin{pmatrix}
     & 0.8  & 0.05 & 0.1 & 0.05 \\
     & 0.15  & 0.5  & 0.25 & 0.1\\
     & 0.1 & 0.05  & 0.25 & 0.6
    \end{pmatrix}.
\end{equation}

For the coefficients corresponding to the effect of $x$, I chose values from a set $\{-1,0,1\}$, whereas for the autoregressive and feedback effects coefficients were from a set  $\{-0.5,0,0.5\}$. For further details, see the replication codes on Github.

I will first compare estimation of this model using the gradient-based numerical optimization algorithm L-BFGS \citep{LBFGS}, and EM-L-BFGS, where the EM algorithm is run for a maximum of 100 iterations before switching to L-BFGS (plain EM was omitted due to computational cost, but see Section D of the supplementary material). The L-BFGS used the implementation provided by the NLopt library \citep{NLopt}, with analytical gradients (see Section C of the supplementary material).

I estimated the model 400 times using both optimization algorithms, with varying the number of estimated hidden states $S=2,3,4$. For each replication, initial values for $\eta$s were sampled from $N(0, 4)$ distribution based on Maximin Latin hypercube sampling \citep{Johnson1990} using R package \texttt{lhs} \citep{lhs}. Exception to this was the expected values of $\eta$s corresponding to $A$, which were set so that $\mathbb{E}(a_{ij}) = 1-0.05(S-1)$ for $i=j$ and 0.05 otherwise. This was done to reflect the fact that often we have some prior idea on the stability of the state trajectories. I also tested the effect of L2 regularization by varying $\lambda \in {0, 0.1}$ in the penalized log-likelihood $\log L(\eta) - \frac{\lambda}{2}| \eta |_2^2$, which can alleviate numerical issues due to extreme probabilities.

I defined success in converging to the global optimum $\max(\log L)$ as achieving a log-likelihood $\log L$ value that satisfies $\lvert \log L - \max(\log L) \rvert < 10^{-5} |\max( \log L)|$, across all replications and estimation methods for a given $S$ and $\lambda$. 

\autoref{tab:optimization_fanhmm} shows the success rates and average computation times for L-BFGS and EM-L-BFGS. Without regularization, both methods, but especially L-BFGS, struggles at converging to the global optimum. The regularization $\lambda = 0.1$ improves the convergence, except when overestimating the number of hidden states.

\begin{table}[!ht]
\centering
\begin{tabular}{llccccc}
 & & \multicolumn{2}{c}{$\lambda = 0$} && \multicolumn{2}{c}{$\lambda = 0.1$}\\
  \cline{3-4} \cline{6-7} \\[-1em] 
S & Method & Success (\%) & Time (s) && Success (\%) & Time (s) \\
\hline
2 & L-BFGS      & 5   & 37    && 40  & 68 \\
  & EM-L-BFGS   & 39  & 294   && 49  & 296 \\
3 & L-BFGS      & 1   & 188   && 43  & 271 \\
  & EM-L-BFGS   & 10  & 769   && 45  & 772 \\
4 & L-BFGS      & 0   & 538   && 1   & 561 \\
  & EM-L-BFGS   & 1   & 1730  && 0   & 1557 \\
\end{tabular}
\caption{Performance comparison of different estimation methods across varying number of estimated number of hidden states $S$ and regularization $\lambda$ for three-state FAN-HMM.}
\label{tab:optimization_fanhmm}
\end{table}

When examining estimates of the average causal effect (ACE) $P(y_{50} = 2 \mid \operatorname{do}(x_{50} = 1)) - P(y_{50} = 2 \mid \operatorname{do}(x_{50}) = 0)$ in \autoref{fig:ace_y2}, we observe that the model with $S=2$ exhibits clear bias relative to the estimate obtained from the correctly specified model with $S=3$. In contrast, with $S=4$, the estimated ACE fluctuates relatively little around the correct value, despite \autoref{tab:optimization_fanhmm} indicating convergence to the global optimum in fewer than 1\% of cases. This behavior reflects a typical feature of mixture and latent variable models in general: When $S$ is underestimated, the model is unable to adequately capture the time-varying heterogeneity of individuals, whereas overspecification of $S=4$ induces redundant state partitions which can complicate model convergence and parameter inference, but its impact on target estimands such as the ACE can remain negligible.

\begin{figure}[!ht]
    \centering
    \includegraphics[width=1\linewidth]{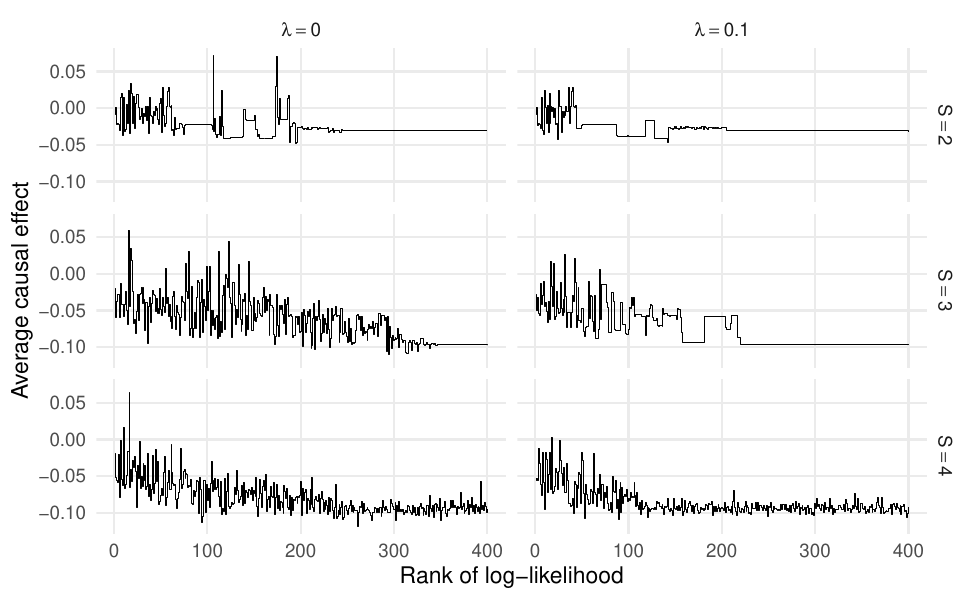}
    \caption{Estimates of $P(y_{50} = 2 \mid \operatorname{do}(x_{50} = 1)) - P(y_{50} = 2 \mid \operatorname{do}(x_{50}) = 0)$ from 400 initializations of three-state FAN-HMM, using EM-L-BFGS, ordered by the log-likelihood values. Rows indicate the number of hidden states $S$, while columns refer to the regularization parameter $\lambda$.}
    \label{fig:ace_y2}
\end{figure}

To further test the estimation of causal effects, I simulated 100,000 sequences from the true model, which I used to compute the ground truth for the ACEs $P(y_t = m \mid \operatorname{do}(x_{46:50} = 1)) - P(y_t = m \mid \operatorname{do}(x_{46:50} = 0))$ for $t=46,\ldots, 50$, $m=1,\ldots,4$. I then simulated 1000 datasets from this model, and for each replication, estimated the model using L-BFGS with $\lambda = 0.1$ and 10 random initializations of the parameters from $N(0, 4)$ distribution to improve the likelihood of finding the global optimum. Then, based on the estimated model, I estimated ACEs with 90\% nonparametric bootstrap confidence intervals of ACEs to test the how close to the nominal coverage the obtained intervals are. For bootstrap, I used only 50 replicates, as the main variation in the covarage probabilities in this experiment is due to the different replicates of data. Naturally for actual applications the number of bootstrap samples should be larger.

\autoref{fig:rmsecoverage} shows how underestimation of number of states $S$ leads to increased root mean square error (RMSE) and confidence intervals with poor coverage, whereas overestimating of $S$ has only a small effect compared to the correctly specified $S$ even though it is likely that the models with $S=4$ converged poorly. Therefore, while correctly specifying $S$ is important for example in case where the interest is in state-specific conditional causal effects, this suggests that when the number of hidden states is not known a priori and the interest causal effects without stratifying by state, it is best to err on the side of overestimation rather than underestimation of the number of hidden states.

\begin{figure}[!ht]
    \centering
    \includegraphics[width=1\linewidth]{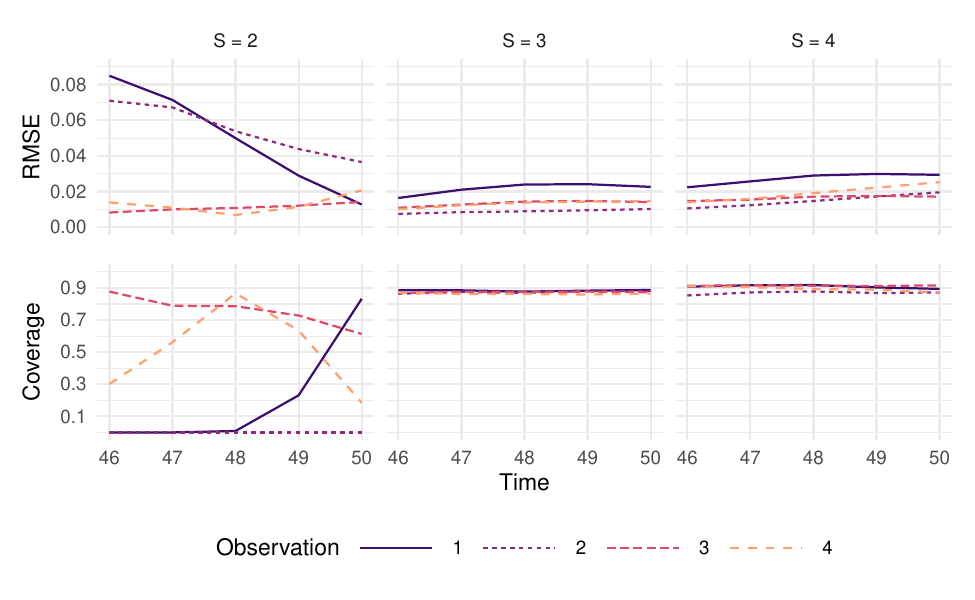}
    \caption{The root mean square error and the coverage of 90\% confidence intervals of the estimated average causal effects from 1000 replications of three-state FAN-HMM.}
    \label{fig:rmsecoverage}
\end{figure}

\section{Application to parental leave reform}

As an empirical illustration of FAN-HMM, I study how the 2013 parental leave reform in Finland affected fathers' parental leave uptake in workplaces in 2010-2017. Fathers' individual entitlements in Finland included of a three-week birth-related paid paternity leave, and a six-week father's quota of paid parental leave. The categorical outcome variable is the leave taking of a father, coded as "no leave", "paternity leave", or "at least father's quota". The main change in the 2013 reform was that the father's quota was no longer conditional on the mother’s consent or shortening her leave. For more details on the context of the Finnish parental leave, see for example \citep{Helske2024_peer} and references therein. 

The data used here are originally sourced from administrative registers compiled by Statistics Finland from the Social Insurance Institution of Finland and the Finnish Institute for Health and Welfare, and was further processed as part of the study by \citet{Helske2024_peer}. Due to data protection laws and regulations, the data are only available from the data holders upon reasonable request pending their permission. However, for illustrative purposes, a small synthetic dataset based on the estimated model and the key characteristics of the original data is available on Github.

I consider 118,469 births for fathers in 12,819 Finnish workplaces with at least five births between 2010 and 2017, up to a maximum of 20 births per workplace. Thus, each observation sequence consists of all fathers having a new child between 2010 and 2017 in a specific workplace, ordered according to the birth date of the children, i.e., the time variable is the rank of the birth in a workplace since January 1 2010. Simple observational level autoregression is likely unable to fully capture the temporal dependencies between fathers due to cumulative effects of changing workplace culture, and as the fathers' parental leave decisions are unevenly spaced in terms of calendar time. Instead of incorporating complex lag-structures and information about time gap between births, I assume that the observations depend also on a latent Markov process capturing time-varying workplace culture, leading to the FAN-HMM.

Using the Wilkinson-Rogers notation \citep{Wilkinson1973}, I assume that for a father $j$ in workplace $i$, the linear predictors for leave taking $y_{j,i}$ and state $z_{j,i}$, with coefficients depending on the current state $z_{j,i}$, are given by 
\begin{equation*}
\begin{aligned}
y_{j,i} &= y_{j-1,i} + x_{j,i} + o_{j,i} + y_{j-1,i} : s_{j,i} + x_{j,i} : o_{j,i},\\
z_{j,i} &= y_{j-1,i} + y_{j-1,i}:x_{j-1,i} + y_{j-1,i}:o_{j-1,i},
\end{aligned}
\end{equation*}
where $a:b$ denotes the interaction effects of $a$ and $b$, $x_{j,i}$ is the reform indicator, $o_{j,i}$ the father's occupational skill level coded as "lower-skill" and "higher-skill", and $s_{j,i}$ is $1$ if $o_{j,i} = o_{j-1,i}$ and 0 otherwise. Additionally, I assume that the initial state probabilities depends on $x_{1,i}$. The corresponding FAN-HMM accounts for the direct effect of reform on the leave probabilities of father $j$ in workplace $i$ ($x_j \to y_j$), as well as indirect effects via $z_j$. Father $j$ might set an example to father $j+1$ ($y_j \to y_{j+1}$), and also change the workplace culture and leave taking behavior of subsequent fathers ($y_j \to z_{j+1} \to y_{j+1}$). The occupations of the father $j$ is likely not only linked to the base level of the leave uptake probabilities, but affects also on how sensitive the workplace culture and father $j+1$ are to the choices of the father $i$.

I estimated the model with $\lambda = 0.1$ and 100 initializations both with $S=2$ and $S=3$ hidden states, and chose the latter based on insights from the simulation experiments. This was also supported by the AIC and BIC measures (with differences 770 and 235, respectively, in favor of $S=3$).

\autoref{fig:emissions} shows the marginal emission probabilities from the estimated model, conditional on the year of birth of the child. State 1 is characterized by high probability of fathers using the quota (about 60\% before 2013 and 70\% after), whereas fathers in workplaces at State 2 tend take mostly paternity leave, while State 3 is mostly characterised by high probability of not taking leave at all (30\% change before 2013, and about 35\% after). In all states, the likelihood of using quota or more increases considerably in 2013, mostly at the cost of decreased probability of using only paternity leave, except in the State 3 where also the probability of not taking any leave increases in 2013. 

The yearly marginal state probabilities are shown in \autoref{fig:states}. Probability of a workplace being in State 1 increases almost linearly over time, while probability of being at State 2 also increases before 2013, before starting to slowly decrease. Probability of State 3 decreases dramatically from over 40\% to 15-20\% in 2013, after which it stays relatively constant. Marginal transition probabilities available in section E of the supplementary material show that the states 1 and 2 have high self-transition probabilities both before and after reform (around 95\% and 92\% in 2010, 98\% and 95\% in 2017), while self-transition probabilities for State 3 in 2010 is only 56\%, but 91\% in 2017. 

 \begin{figure}[!ht]
    \centering
    \includegraphics[width=1\linewidth]{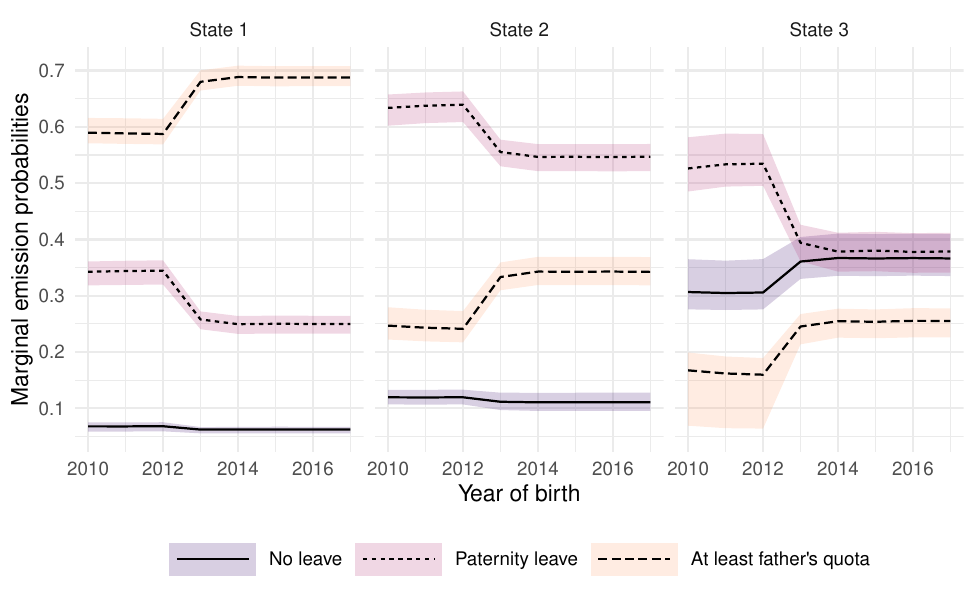}
    \caption{Average marginal emission probabilities and their 95\% confidence intervals by the child's year of birth.}
    \label{fig:emissions}
\end{figure}
\begin{figure}[!ht]
    \centering
    \includegraphics[width=1\linewidth]{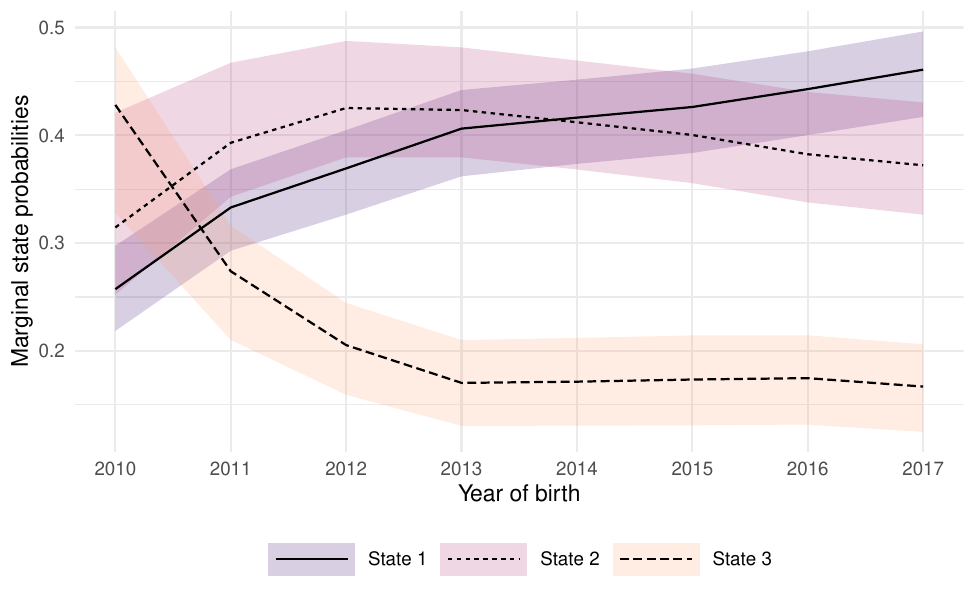}
    \caption{Average marginal state probabilities and their 95\% confidence intervals by the child's year of birth.}
    \label{fig:states}
\end{figure}

\autoref{fig:leaves} shows the causal effects $P(y_t \mid \operatorname{do}(x_{2013:2017} = 1), w_t) - P(y_t \mid \operatorname{do}(x_{2013:2017} = 0), w_t)$, $t=0,\ldots,10$ for the two occupation levels $w$, where $t$ is the rank of a father in a workplace since the 2013 reform. The reform effect is clear: the probability of using quota increases by 13 percentage points (ppts) (95\% confidence interval $(12,14)$ ppts) for first fathers in the lower-skill occupation category, and 6 $(4, 7)$ ppts for higher-skill occupations. The probability of using only paternity leave slowly decreases when more and more fathers are exposed to the reform. There is a minor positive effect in on the probability of not taking any leave at all, reaching at most 1.7 and 1.3 ppts for lower and higher occupations, with 95\% CIs $(0.9, 2.4)$ and $(0.5, 2.1)$, respectively, but these effects start to diminish after $t=4$. After this initial phase, the positive effect on father's quota use begins to increase, reaching 17 $(14, 19)$ and 8 $(5, 10)$ ppts for fathers with lower-skill and higher skill occupations, respectively. Section E of the online supplement has additional figures of the reform effect.

\begin{figure}[!ht]
    \centering
    \includegraphics[width=1\linewidth]{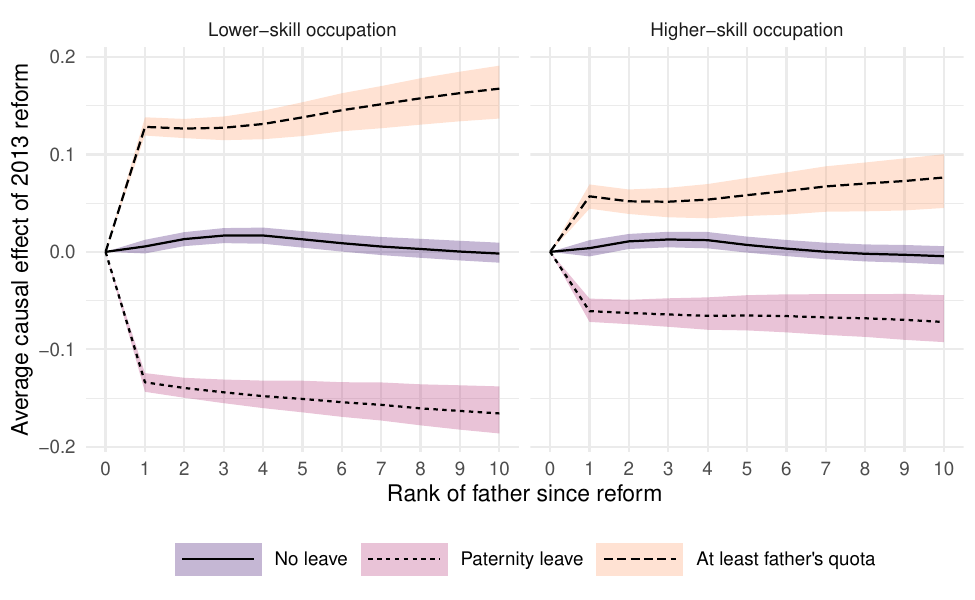}
    \caption{Average causal effects and 95\% confidence intervals of the Finnish 2013 parental leave reform, conditional on the occupation level.}
    \label{fig:leaves}
\end{figure}

\section{Discussion}

Feedback-augmented non-homogeneous HMMs provide a novel approach for observational causal inference in case of time-varying heterogeneity when the interest is in the causal effects of observables on the outcomes. The standard assumption of HMMs, that individuals' latent trajectories can be adequately represented by sequences of a relatively few discrete latent states, might not always be plausible. While the autoregressive and feedback dependencies in FAN-HMMs provide additional flexibility, they of course are not a panacea for capturing complex time-varying heterogeneity in general. While focusing here on maximum likelihood estimation, further research on fully Bayesian approaches for estimating FAN-HMMs could be valuable as it would allow easier augmentation of the model with additional features such as individual-level random effects.

Although estimating FAN-HMMs, and HMMs in general, can be computationally demanding, in the case of categorical outcomes, the subsequent computation of the causal effects can be performed efficiently by minor modifications of the standard forward algorithm. Due to the marginalization of latent states, the estimates can be robust to overestimation of the number of latent states, while state-specific effects can provide additional insights on the heterogeneity of the causal effects. 

There exists various approaches for choosing the number of hidden states of HMM \citep{Zucchini2009, Visser2022}. \citet{Pohle2017} gives several practical suggestions, and shows that when the true DGP does not fully match standard HMM, state selection based on information criteria tends to favor extra states which try to capture the neglected aspects of the data-generating process. \citet{Pohle2017} also points the importance of considering whether the aim is to obtain interpretable hidden states and corresponding emission distributions, or if the model is to be mainly used for forecasting where larger number of states may provide more accurate predictions. The causal inference task considered here is more akin to the latter (unless the interest is in state-specific causal effects). Simulation experiments suggest that when the number of hidden states is unknown, it might be reasonable to lean toward overestimation, which, at least in the experiments conducted here, is constrained by the estimability of models with too many states. This increased difficulty in finding the same mode under repeated estimation with multiple initial values could be used as an additional heuristic strategy for selecting the number of states. Similar reasoning was also used by \citet{Pasanen2024} in the context of Bayesian spatio-temporal HMM, where the aim was to find few interpretable latent states capturing the dynamics of a process varying smoothly over time.

The empirical application of the FAN-HMM for studying a parental leave reform illustrated how FAN-HMM allows modeling interventions having multiple causal pathways with cumulative long-term effects. Naturally, the justification of causal interpretation of these results relies on the assumption that there were no other events in 2013 which could have affected the leave taking of fathers, and that the FAN-HMM sufficiently captured the relevant aspects of the data generating process. While additional data for example on the partners' education level could reduce the uncertainty and the risk of model misspecification, I kept the model relatively simple for illustrative purposes. With these caveats, I estimated clear positive effect of the reform on the amount of leave taken, with only a minor decrease on the probability of taking no leave at all. Overall, fathers in occupations with lower skill requirements were more affected by the reform, likely partially due to the fact that highly educated fathers (likely in higher-skill occupations) were already using more parental leaves before reform \citep{Helske2024_peer}.  

\section{Acknowledgments}

\if1\blind
{
This work was supported by the INVEST Research Flagship Centre and the Research Council of Finland under grants 331817, 355153, and 345546. I wish to acknowledge CSC – IT Center for Science, Finland, for computational resources. I wish to thank Satu Helske for helpful comments and Simon Chapman for providing the original parental leave data. ChatGPT-4o was used to assist in refining the manuscript's flow and grammar, with all content subsequently reviewed and edited by the author. There are no competing interests to declare.
} \fi

\if0\blind
{
ChatGPT-4o was used to assist in refining the manuscript's flow and grammar, with all content subsequently reviewed and edited by the author. There are no competing interests to declare.
} \fi

\appendix

\section{Parameterization of coefficients of NHMMs}\label{app:gammas}

To keep the softmax functions in Equation (1) identifiable, a typical approach is to set one of the rows of the coefficient matrices $\boldsymbol{\gamma}$ to zero. However, I use an alternative approach in which I define $\boldsymbol{\gamma} = \mathbf{Q}\boldsymbol{\eta}$, where $\boldsymbol{\eta}$ is $S-1 \times K$ matrix of working parameters, $\mathbf{Q}$ is an $S\times S - 1$ orthogonal matrix obtained by taking the first $S - 1$ columns of the $\mathbf{Q}$ matrix from the QR decomposition of a $S \times S$ matrix $\mathbf{D}$. The first $S-1$ rows of $\mathbf{D}$ contain identity matrix, while the last row is $-\mathbf{1}$. Thus, the matrix $\mathbf{Q}$ projects $\boldsymbol{\eta}$ to $\boldsymbol{\gamma}$ so that each column of $\boldsymbol{\gamma}$ sums to zero, and the columns are orthogonal (for $\boldsymbol{\gamma}^B_s$, replace $S$ with $M$). The main reason for this sum-to-zero parameterization is that it makes handling of label-switching easier by allowing the permutation of the rows of $\boldsymbol{\gamma}$, see \autoref{app:MLE}. In addition, the orthogonality transformation can potentially improve the performance of the numerical optimization algorithms used in the maximum likelihood estimation by making the $\gamma$s equicorrelated with smaller (negative) pairwise correlations compared  to a simple sum-to-zero construction $\boldsymbol{\gamma} = \mathbf{D}\boldsymbol{\eta}$, where the last row of $\gamma$ has strong (negative) correlations with the other rows. A similar approach has also been recommended in the Bayesian context of general sum-to-zero vectors, as this makes it easier to define symmetric priors and improves the efficiency of gradient-based sampling algorithms. 

\section{Parameter estimation of FAN-HMMs}\label{app:MLE}

The maximum likelihood estimation of HMMs has been traditionally performed using the Baum-Welch algorithm \citep{Baum1970}, a specific instance of more general EM algorithm \citep{Dempster1977}. However, an alternative approach of direct numerical maximization (DNM) of the marginal likelihood was already proposed by \citet{Levinson1983}, and later recommended by many authors \citep{Altman2005,Turner2008,Zucchini2009}. The strengths and weaknesses of both EM and DNM have been discussed also in \citep{Cappe2005} and \citep{Visser2022}. A key advantage of DNM is its faster convergence rate compared to EM. \citet{Zucchini2009} notes the EM algorithm relies on outputs from both the forward and backward algorithms \citep{Rabiner1989}, whereas the marginal likelihood in DNM can be computed using only the forward algorithm. However, gradient-based optimization methods for DNM still require the backward algorithm to compute analytical gradients. On the other hand, in case of NHMMs, the M-step of the EM algorithm cannot be solved in closed form, necessitating the use of numerical optimization methods even within the EM framework.

In addition to computation time, an important consideration in choosing an optimization method is its ability to converge to the global optimum, as the likelihood of HMMs is multimodal \citep{Rabiner1989, Cappe2005, Zucchini2009, Visser2022}. Often the discussion of multimodality of HMMs focuses on label-switching problem which refers to the fact that the latent states can be ordered arbitrary without changing the likelihood value. Much work has been done to mitigate the issues of label-switching especially in Bayesian context, see for example \citep{Stephens2002, Jasra2005}. From a maximum likelihood perspective, the multimodality due to the label-switching is typically not an issue during the parameter estimation, although it can still cause problems when performing bootstrap sampling for uncertainty estimates due to permutations of the maximum likelihood estimates. However, often the likelihood surface contains also genuine local optima which makes the estimation of HMMs challenging.

\citet{Bacri2023} compared various numerical optimization algorithms for estimating HMMs with Poisson and Gaussian responses using automatic differentiation with Template Model Builder \citep{Kristensen2016}. \citet{Bacri2023} found that most methods they compared performed similarly in terms of finding the global optimum when algorithms were initialized with random values. However, their criteria for identifying the global optimum was relatively lenient, allowing a 5\% tolerance in log-likelihood values. Moreover, \citet{Bacri2023} defined the global optimum by initializing all algorithms with the true parameter values and taking the median of the resulting optima rather than the maximum, which does not accurately represent the global optimum unless all algorithms converge to same value.

The multimodality of the likelihood is likely greater when the model is misspecified, which is essentially always true when working for example with data related to complex social processes. While FAN-NHMMs and NHMMs can in principle alleviate the degree of model misspecification compared to HMMs by capturing heterogeneity of individuals using the covariate information, the general issue of multimodality remains a problem of (non-homogeneous) hidden Markov models.

In addition to point estimates, we naturally want some measures of uncertainty for the model parameters and functions of these. \citet{Zucchini2009} and \citet{Visser2022} discuss obtaining uncertainty estimates using parametric bootstrap and inverse of Hessian. \citet{Bacri2023} uses automatic differentiation to obtain Hessian and subsequent Delta method to compute confidence intervals for the functions of model parameters. However, the regularity conditions required for the asymptotic normality of MLE are often violated in the context of HMMs and consequently, uncertainty estimates based on the Hessian can be unreliable (see, e.g., Section 3.6.1 in \citep{Zucchini2009} and references therein). On the other hand, bootstrap approaches are computationally demanding, and parametric bootstrap relies on the assumption that the model is correctly specified. Given a NHMM with multiple observed sequences, a nonparametric bootstrap with resampling of sequences can be expected to be more robust to the misspecification of the number of hidden states compared to the parametric bootstrap. In order to deal with the label-switching between bootstrap replications, given the sum-to-zero constrain on $\boldsymbol{\gamma}$ and maximum likelihood estimate $\boldsymbol{\hat\gamma}$, we can permute the states of each bootstrap replicate $j$ for example by finding the permutation that minimizes $\|\boldsymbol{\gamma} - \boldsymbol{\gamma}\|_2$. For typical problems with small number of states such as $S < 4$, an exhaustive search is reasonable, while in more general case this can be done efficiently using Hungarian algorithm \citep{Kuhn1955}.

For Bayesian inference of HMMs, multimodality can affect the performance of the posterior sampling via Markov chain Monte Carlo (MCMC) methods and render many typical convergence measures unusable. Some MCMC samplers such as Gibbs tend to sample around single posterior mode if the modes are well-separated, giving a false impression of unimodal posterior, whereas more efficient samplers such as Hamiltonian Monte Carlo may sometimes be able to explore multiple nearby modes, but still in general struggle to fully sample multimodal posteriors. On the other hand, methods designed for multimodal posteriors such as parallel tempering \citep{Geyer1991} typically require some problem-specific tuning. While focusing on the posterior around single highest (global) mode can be reasonable in practice, it is somewhat unsatisfactory solution from a fully Bayesian perspective as the results might do not reflect the full posterior uncertainty of the problem \citep{Jasra2005}. On the other hand, often the interest in using HMMs is in interpreting, visualizing and summarizing the latent dynamics where the latent states are labeled based on the characteristics of estimated transition and emission probabilities. The benefit of using maximum likelihood, or a maximum a posteriori, estimates, is that by focusing on single mode we get more easily interpretable hidden states compared to for example fully Bayesian (multimodal) parameter estimates which are challenging to summarise and interpret compactly.

\section{Gradients for NHMMs}\label{app:gradients}

As in the case of homogeneous HMMs, the log-likelihood of NHMM can be obtained using the forward algorithm, given the model parameters $\boldsymbol{\pi}_i$, $\mathbf{A}_{i,t}$, and $\mathbf{B}_{i,t}$, $i = 1,\ldots, N$, $t=1,\ldots, T_i$. To simplify the notation, I drop the index $i$ in the following; the final log-likelihood and the gradients are the sum of the individual contributions. Denote the likelihood function $L(\eta^\pi, \eta^\pi, \eta^B)$ as $L$. Extending results of \citep{Levinson1983} to time-varying emission and transition matrices, we have
\begin{align*}
\frac{\partial \log L}{\partial \pi} =& \frac{\mathbf{B}_1(:, y_1) \circ \beta_1}{L},\\
\frac{\partial \log L}{\partial \mathbf{A}_t(s,:)} =& \frac{\alpha_{t}(s)\mathbf{B}_{t+1}(:, y_{t+1}) \circ \beta_{t+1}}{L},\quad s = 1,\ldots, S, \quad t = 2, \ldots, T,\\
\frac{\partial \log L}{\partial \mathbf{B}_t(s,:)} =& \mathbf{e}_{y_t} \circ
\frac{\textrm{I}(t > 1)\sum_{s=1}^S \alpha_{t}(s) \mathbf{A}_t(s,:) \circ \beta_{t+1} + \textrm{I}(t = 1)\pi \circ \beta_1}{L},\\
&\quad s = 1,\ldots, S, \quad t = 1, \ldots, T,
\end{align*}
where $\alpha$ and $\beta$ are the standard forward and backward probabilities \citep{Rabiner1989}, $\mathbf{A}_t(s,:)$ and $\mathbf{B}_t(:,m)$ denote row $s$ of $\mathbf{A}^t$ and column $m$ of $\mathbf{B}^t$ respectively, and $\mathbf{e}_{i}$ is the basis vector where $e(j) = 1$ if $j=i$ and $0$ otherwise.

Using the chain rule we have
\begin{align*}
\frac{\partial \log L}{\partial \eta^\pi} =& Q_S^\top (\textrm{diag}(\boldsymbol{\pi}) - \boldsymbol{\pi}\boldsymbol{\pi}^\top)\frac{\partial \log L}{\partial \boldsymbol{\pi}} x^\pi,\\
\frac{\partial \log L}{\partial \eta^A_s} =& Q_S^\top\sum_{t=2}^{T} \left[\left(\textrm{diag}(\mathbf{A}_t(s,:) - \mathbf{A}_t(s,:)^\top \mathbf{A}_t(s,:)\right) \frac{\partial \log L}{\partial \mathbf{A}_t(s,:)}x_t^A\right],\quad s = 1,\ldots, S,\\
\frac{\partial \log L}{\partial \eta^B_s} =& Q_M^\top \sum_{t=1}^T \left[\left(\textrm{diag}(\mathbf{B}_t(s,:) - \mathbf{B}_t(s,:)^\top \mathbf{B}_t(s,:)\right)\frac{\partial \log L}{\partial \mathbf{B}_t(s,:)}x_t^B\right],
\quad s = 1,\ldots, S,
\end{align*}
where $x^\pi, x_t^A$ and $x_t^B$ are the row vectors of covariates, $\eta^A_s$ and $\eta^B_s$ are working parameters related to row $s$ of matrices $\mathbf{A}$ and $\mathbf{B}$ respectively, and $Q_S$ and $Q_M$ are $S \times S - 1$ and $M \times M - 1$ transformation matrices defined in \autoref{app:gammas}.

\subsection{Gradients for the M-step of EM algorithm}

Given the expected counts $\mathbf{E}^\pi$, $\mathbf{E}_t^A$, and $\mathbf{E}_t^B$ from the previous E-step of the EM algorithm, we maximize the expected log-likelihood functions
\begin{align*}
\mathbb{E}[\log L(\eta^\pi)] =& \log(\pi)^\top \mathbf{E}^\pi,\\
\mathbb{E}[\log L(\eta^A_s)] =& \sum_{t=2}^{T}\log(\mathbf{A}_t(s,:)) \mathbf{E}^A_t(s,:)^\top,\quad s = 1,\ldots, S,\\
\mathbb{E}[\log L(\eta^B_s)] =& \sum_{t=1}^{T}\log(\mathbf{B}_t(s,y_t)) \mathbf{E}^B_t(s,y_t),\quad s = 1,\ldots, S.
\end{align*}

The corresponding gradients are
\begin{align*}
\frac{\partial\mathbb{E}[\log L(\eta^\pi)]}{\partial \eta^\pi} =& Q_S^\top\left[\mathbf{E}^\pi - \mathbf{1}_S^\top \mathbf{E}^\pi \pi\right] x^\pi,\\
\frac{\partial\mathbb{E}[\log L(\eta^A_s)]}{\partial \eta^A_s} =& Q_S^\top\sum_{t=2}^{T}\left[(\mathbf{E}^A_t(s,:) - \mathbf{E}^A_{t}(s,:)\mathbf{1}_S \mathbf{A}_t(s,:))^\top x^A_t\right],\quad s = 1,\ldots, S,\\
\frac{\partial\mathbb{E}[\log L(\eta^B_s)]}{\partial \eta^B_s} =& Q_M^\top\sum_{t=1}^{T}\left[\mathbf{E}_t^B(s,y_t) (\mathbf{e}_{y_t} - \mathbf{B}_t(s,:)^\top )x^B_t\right],\quad s = 1,\ldots, S.
\end{align*}

\section{Additional simulation experiments}\label{app:hmm_s4}

Here I compare estimation of standard HMMs using EM algorithm, gradient based numerical optimization algorithm L-BFGS \citep{LBFGS}, and their hybrid version EM-L-BFGS.

For the experiment, I generated synthetic data using three-state HMM defined as
\begin{equation}\label{eq:hmm3}
\pi = \begin{pmatrix}
     0.8  \\
     0.1 \\
     0.1
    \end{pmatrix},
A = \begin{pmatrix}
     & 0.85  & 0.1 & 0.05 \\
     & 0.1 & 0.8 & 0.1\\
     & 0.05 & 0.05 & 0.9
    \end{pmatrix},
B = \begin{pmatrix}
     & 0.8  & 0.05 & 0.1 & 0.05 \\
     & 0.15  & 0.5  & 0.25 & 0.1\\
     & 0.1 & 0.05  & 0.25 & 0.6
    \end{pmatrix}.
\end{equation}
I simulated $N=1000$ sequences of length $T=50$ from this model, and estimated the model 400 times using all three optimization algorithms, with varying the number of estimated hidden states $S=2,3,4$. As in the simulation experiments in the main paper, for each replication, initial values in terms of working parameters $\eta$ were sampled from zero-mean normal distribution with $\sigma^2=4$, except that the expected values of $\eta$s corresponding to $A$ were set so that $\mathbb{E}(a_{ij}) = 1-0.05(S-1)$ for $i=j$ and 0.05 otherwise.

\autoref{tab:optimization} shows the success rate and computation time in seconds for each configuration. First, we can see that all methods struggle somewhat in converging to global mode when the number of states is overestimated, especially without any regularisation. When model is correctly specified, EM and EM-L-BFGS perform relatively similar in terms of success rate, with L-BFGS performing slightly worse. Without regularisation, the computation speed of L-BFGS and the hybrid EM-L-BFGS is comparable but both methods using EM algorithm slow down with regularisation as in that case the M-step need to rely on numerical optimization. The plain EM slows down significantly in the problematic case of $S=4$, especially with regularization.

\begin{table}[!ht]
\centering
\begin{tabular}{llccccc}
 & & \multicolumn{2}{c}{$\lambda = 0$} && \multicolumn{2}{c}{$\lambda = 0.1$}\\
  \cline{3-4} \cline{6-7} \\[-1em] 
S & Method & Success (\%) & Time (s) && Success (\%) & Time (s) \\
\hline
2 & L-BFGS      & 89    & 4  && 99.8  & 4 \\
  & EM          & 99.5  & 3  && 99.8  & 62 \\
  & EM-L-BFGS   & 99.5  & 3  && 99.5  & 41 \\
3 & L-BFGS      & 86.5  & 9  && 100  & 10 \\
  & EM          & 99.5  & 5  && 99  & 100 \\
  & EM-L-BFGS   & 98.5  & 4  && 99  & 59 \\
4 & L-BFGS      & 11.2  & 24  && 33.5 & 31 \\
  & EM          & 13.5  & 86 && 20.5 & 1571 \\
  & EM-L-BFGS   & 17.2  & 23  && 25.5 & 151 \\
\end{tabular}
\caption{Performance comparison of different estimation methods across varying number of estimated number of hidden states $S$ and regularization $\lambda$ for three-state HMM.}
\label{tab:optimization}
\end{table}

\autoref{fig:max_A_S3} shows how the convergence issues manifest in the actual parameter estimates by plotting the estimated maximum transition probabilities, $\max_{s,r} a_{sr}$ for each initialization, sorted by the rank of the the log-likelihood values. For illustrative purposes, \autoref{fig:max_A_S3} displays only results for EM-L-BFGS method; results for other optimization methods are comparable. The fluctuation of the values in case of $S=4$ indicates practical issues when overestimating the number of hidden states, as the transition probabilities vary several percentage points. While using weak regularization helps in convergence, it does not significantly affect the parameter estimate; the horizontal lines for specific $S$ are aligned, indicating that the estimated transition probabilities remain the same across both cases.

\begin{figure}[!ht]
    \centering
    \includegraphics[width=1\linewidth]{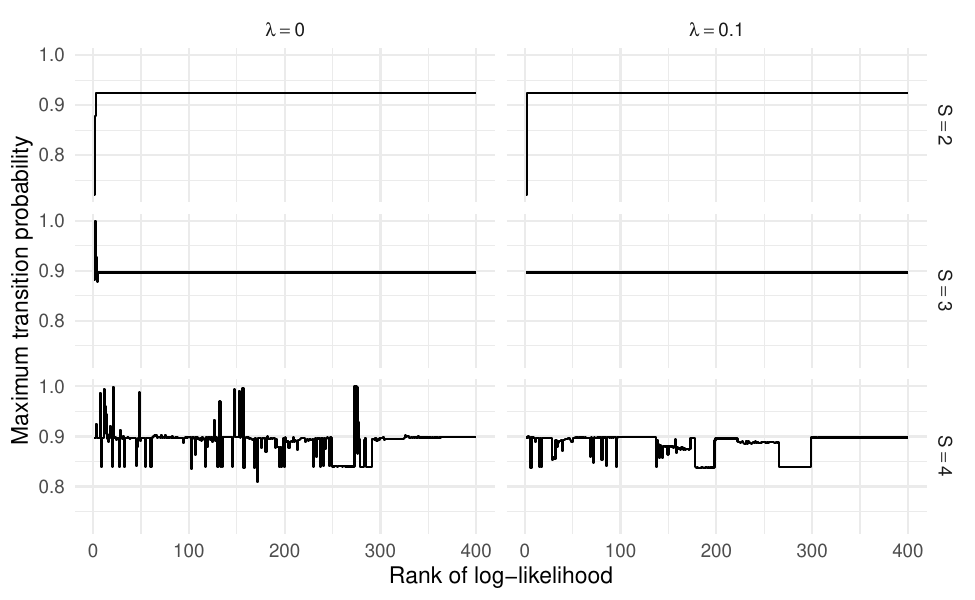}
    \caption{The estimated maximum value of transition probabilities from 400 replications using EM-L-BFGS, ordered by the log-likelihood values. Rows indicate the estimated number of hidden states, while columns refer to the regularization parameter $\lambda$.}
    \label{fig:max_A_S3}
\end{figure}

\autoref{tab:optimization_S4} and \autoref{fig:max_A_S4} shows the results for a case of four-state HMM with $M=6$ observed symbols, again with $N=1000$ and $T=50$, with the true model parameters corresponding to
\[
\pi = \begin{pmatrix}
     0.8  \\
     0.1 \\
     0.05 \\
     0.05
    \end{pmatrix},
A = \begin{pmatrix}
     & 0.8 & 0.1 & 0.05 & 0.05 \\
     & 0.1 & 0.7 & 0.1 & 0.1 \\
     & 0.05 & 0.05 & 0.8 & 0.1 \\
     & 0.025 & 0.05 & 0.025 & 0.9
    \end{pmatrix},
B = \begin{pmatrix}
     & 0.5  & 0.1 & 0.1 & 0.05 & 0.05 & 0.2 \\
     & 0.1  & 0.5 & 0.1 & 0.1 & 0.1 & 0.1\\
     & 0.2 & 0.05 & 0.3 & 0.3 & 0.05 & 0.1 \\
     & 0.05 & 0.05 & 0.1 & 0.3 & 0.4 & 0.1
    \end{pmatrix}.
\]

I left the plain EM algorithm out from the comparisons due to its computational load. Overall these results are in line with the case of three-state HMM and the FAN-HMM experiment: overestimation of hidden states cause severe problems to optimization.

\begin{table}[!ht]
\centering
\begin{tabular}{llccccc}
 & & \multicolumn{2}{c}{$\lambda = 0$} && \multicolumn{2}{c}{$\lambda = 0.1$}\\
  \cline{3-4} \cline{6-7} \\[-1em] 
S & Method & Success (\%) & Time (s) && Success (\%) & Time (s) \\
\hline
2 & L-BFGS      & 78.5  & 5  && 99.8  & 6 \\
  & EM-L-BFGS   & 99.5  & 2  && 99.2  & 37 \\
3 & L-BFGS      & 56.5  & 14  && 71  & 16 \\
  & EM-L-BFGS   & 65.2  & 9  && 64.8  & 108 \\
4 & L-BFGS      & 74.2  & 27  && 99.5 & 27 \\
  & EM-L-BFGS   & 85.5  & 19  && 85.2 & 174 \\
5 & L-BFGS      & 4.5   & 61  && 45.8 & 68 \\
  & EM-L-BFGS   & 6.8   & 52  && 41.8 & 279 \\
6 & L-BFGS      & 0     & 96  && 1 & 100 \\
  & EM-L-BFGS   & 0.5   & 88  && 1.5 & 382 \\
\end{tabular}
\caption{Performance comparison of different estimation methods across varying number of estimated number of hidden states $S$ and regularization $\lambda$ for four-state HMM.}
\label{tab:optimization_S4}
\end{table}

\begin{figure}[!ht]
    \centering
    \includegraphics[width=1\linewidth]{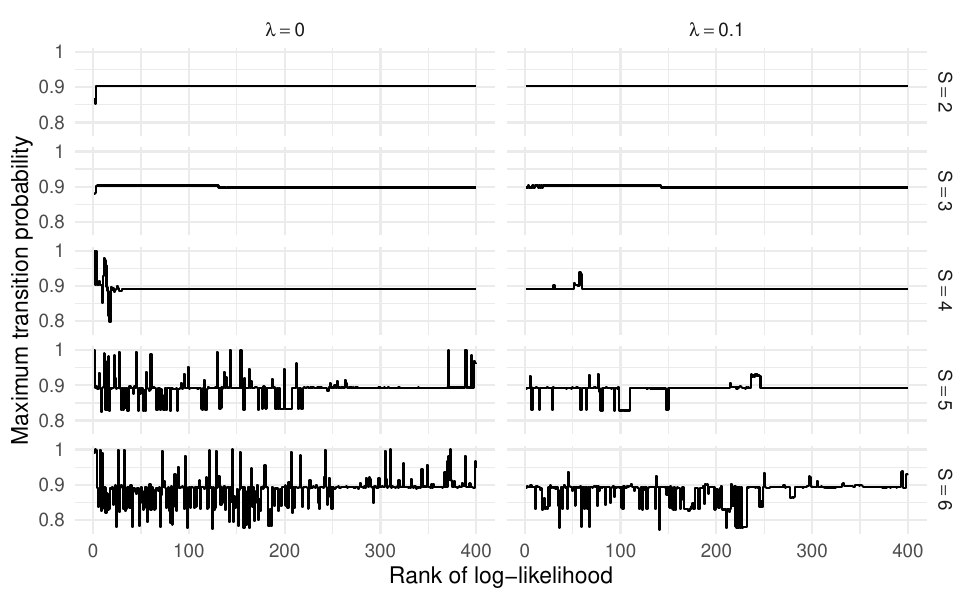}
    \caption{The estimated maximum value of transition probabilities of four-state HMM from 400 initializations using EM-L-BFGS, ordered by the log-likelihood values. Rows indicate the estimated number of hidden states, while columns refer to the regularization parameter $\lambda$.}
    \label{fig:max_A_S4}
\end{figure}

\subsection{Bayesian estimation}\label{app:bayes}

Here, I illustrate the practical challenges of Bayesian HMM estimation using the same four-state HMM as before. The model was estimated using Stan \citep{stan} with 16 parallel chains, each of length 1000 + 1000 iterations, with the first 1000 discarded as warm-up. The prior for the $\eta$ parameters was set as $N(0, \sqrt{10})$, matching the regularisation $\lambda = 0.1$ used in the maximum likelihood case. The model code is available on GitHub. \autoref{tab:hmm_S4_stan} shows the computation times and $\hat{R}$ values \citep{Vehtari2021} for the log-posterior and the maximum transition probabilities (both of which are invariant to label-switching), as well as the bulk and tail effective sample sizes \citep{Vehtari2021} of these quantities.

For $S=3$, $\hat{R}$ values well over 1.01 indicate a lack of convergence, while for the true number of hidden states ($S=4$), $\hat{R}$ values suggest convergence with large effective sample sizes. For $S=5$, the model appears to converge based on $\hat{R}$, but the effective sample sizes are considerably smaller than in the $S=4$ case. In addition, with $S=5$, there were 24 divergent transitions, indicating potential bias due to the sampler's insufficient exploration of the posterior. Furthermore, 1754 transitions hit the maximum tree depth, manifesting as inefficient sampling. This is also reflected in the computation times: compared to the run time of 37 minutes for $S=4$, sampling took almost 12 hours for $S=5$. In summary, misspecification of $S$ leads to various issues in model estimation, as in the maximum likelihood case.

\begin{table}[!ht]
\centering
\begin{tabular}{lcccccccc}
 & & \multicolumn{3}{c}{log-posterior} && \multicolumn{3}{c}{$\max_{s,r} a_{sr}$}\\
  \cline{3-5} \cline{7-9}\\[-1em] 
S & Time (min) & $\hat R$ & ESS-bulk & ESS-tail && $\hat R$ & ESS-bulk & ESS-tail\\
\hline
3 & 32  & 1.39 & 34   & 81   && 1.13 & 78     & 139 \\
4 & 37  & 1.00 & 6173 & 9467 && 1.00 & 16717 & 14864 \\
5 & 717 & 1.01 & 1168 & 2687 && 1.01 & 3714  & 1824
\end{tabular}
\caption{Computation times, $\hat R$, and ESS values of Bayesian four-state HMM with different number of estimated states $S$.}
\label{tab:hmm_S4_stan}
\end{table}

\section{Additional figures for parental leave application}

\autoref{fig:transitions} shows the average marginal yearly transition probabilities. Workplaces in states 1 and 2 have high self-transition probabilities, and transitioning from State 1 to State 3 is very unlikely. Workplaces in State 3 are more likely to transition to two other states. However, in all states, and especially in State 3, the self-transition probabilities increase after 2013.

 \begin{figure}[!ht]
    \centering
    \includegraphics[width=1\linewidth]{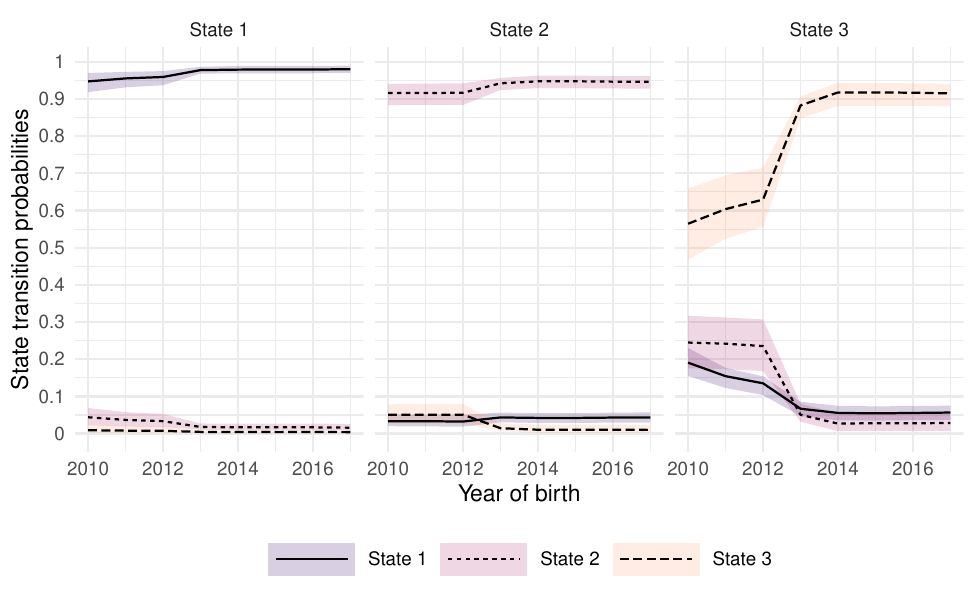}
    \caption{Average marginal state transition probabilities and their 95\% confidence intervals, conditional on the occupation.}
    \label{fig:transitions}
\end{figure}

The average causal effect of the 2013 reform on the emission probabilities in \autoref{fig:ace_emission} show that the probability of quota usage increased in all states due to reform with somewhat varying degrees, and interestingly the probability of not using any leave also increased in State 3, which explain the temporary positive effect of reform on the probability of not taking any leave visible in Figure 7 of the main paper. From \autoref{fig:ace_state} we see that that there is first an increase in the marginal probability of State 3, which then starts to diminish after few fathers. The likelihood of being in State 2 decreases over time, while probability of State 1 increases.

 \begin{figure}[!ht]
    \centering
    \includegraphics[width=1\linewidth]{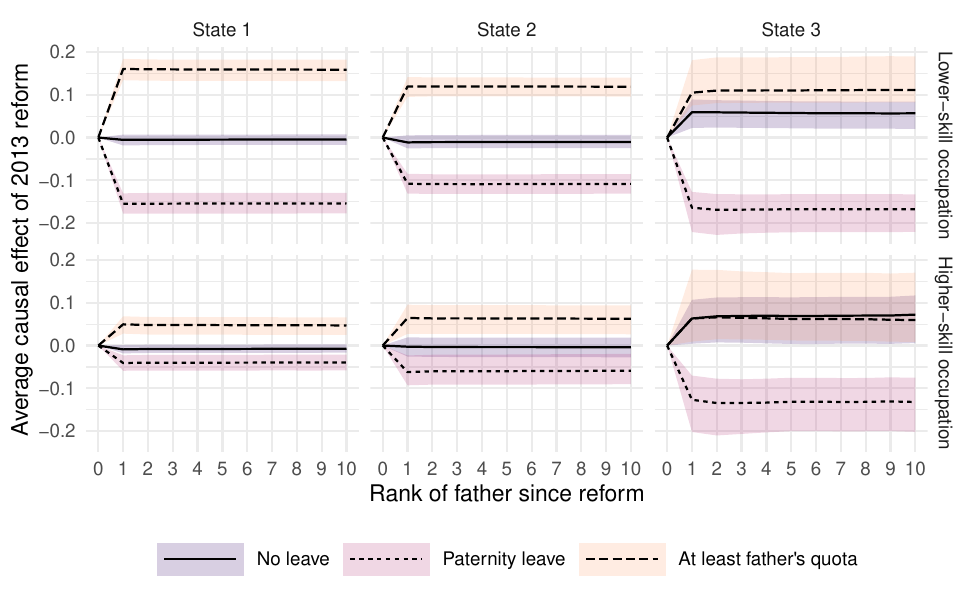}
    \caption{Average causal effect of 2013 reform on the emission probabilities and their 95\% confidence intervals, conditional on the occupation.}
    \label{fig:ace_emission}
\end{figure}
\begin{figure}[!ht]
    \centering
    \includegraphics[width=1\linewidth]{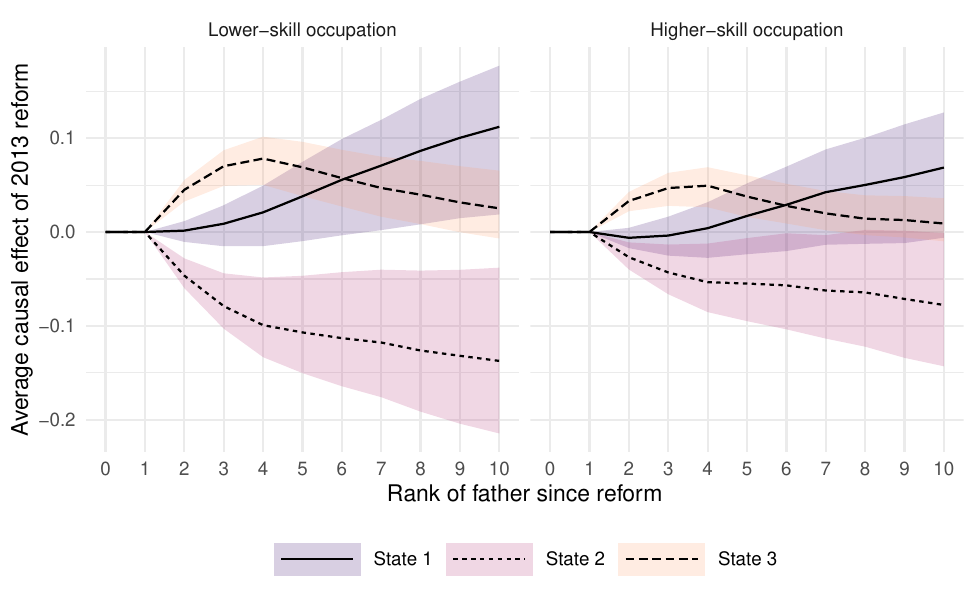}
    \caption{Average causal effect of 2013 reform on the emission probabilities and their 95\% confidence intervals, conditional on the occupation.}
    \label{fig:ace_state}
\end{figure}
\clearpage
\newpage
\bibliographystyle{apalike}
\bibliography{refs}

\end{document}